\documentclass[accepted]{uai2025} % after acceptance, for a revised version; 
\usepackage[american]{babel}
\usepackage{natbib} % has a nice set of citation styles and commands
\usepackage{mathtools} % amsmath with fixes and additions
\usepackage{booktabs} % commands to create good-looking tables
\usepackage{tikz} % nice language for creating drawings and diagrams

\usepackage[commenters={T}]{shortex}

\title{Tuning-Free Coreset Markov Chain Monte Carlo via Hot DoG}

\author[1]{\href{mailto:<naitong.chen@stat.ubc.ca>?Subject=Your UAI 2025 paper}{\color{black}{Naitong~Chen}}{}}
\author[2]{Jonathan~H.~Huggins}
\author[1]{Trevor~Campbell}

\affil[1]{%
    Department of Statistics\\
    University of British Columbia\\
    Vancouver, BC, Canada
}
\affil[2]{%
    Department of Mathematics \& Statistics and Faculty of Computing \& Data Sciences\\
    Boston University\\
    Boston, MA, USA
}
  
\begin{document}
\maketitle

\begin{abstract}
A Bayesian coreset is a small, weighted subset of a data set that replaces the
full data during inference to reduce computational cost.  The state-of-the-art
coreset construction algorithm, \emph{Coreset Markov chain Monte Carlo}
(Coreset MCMC), uses draws from an adaptive Markov chain targeting the coreset
posterior to train the coreset weights via stochastic gradient optimization.
However, the quality of the constructed coreset, and thus the quality of its
posterior approximation, is sensitive to the stochastic optimization learning
rate.  In this work, we propose a learning-rate-free stochastic gradient
optimization procedure, \emph{Hot-start Distance over Gradient} (Hot DoG),
for training coreset weights in Coreset MCMC without user tuning effort.
We provide a theoretical analysis of the convergence of the coreset weights 
produced by Hot DoG.
We also provide empirical results demonstrate that Hot DoG provides higher 
quality posterior
approximations than other learning-rate-free stochastic gradient methods, and
performs competitively to optimally-tuned ADAM.
\end{abstract}

\section{Introduction}
\label{sec:introduction}

\begin{figure}[h]
    \includegraphics[width=\columnwidth]{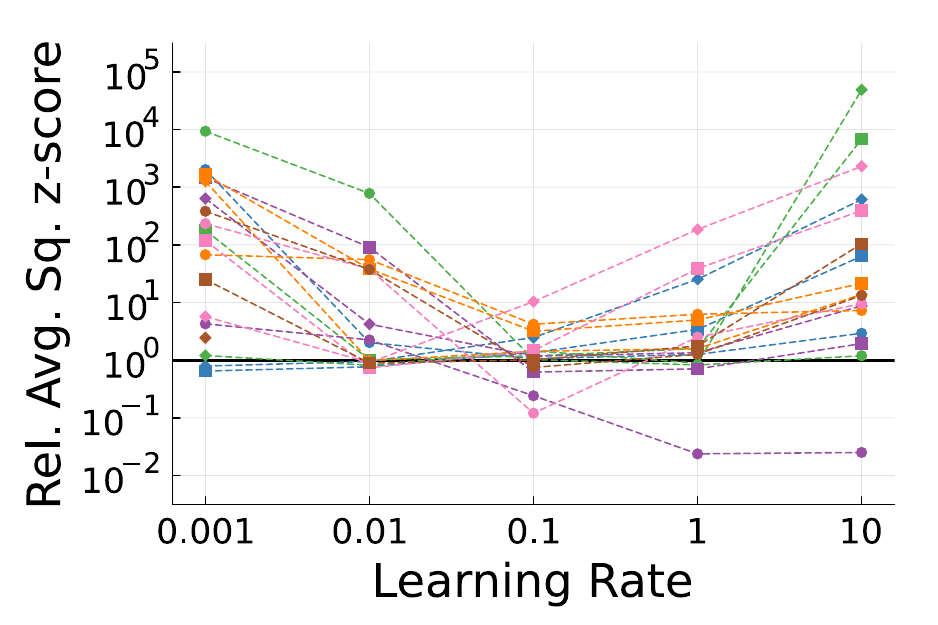}
    \includegraphics[width=\columnwidth]{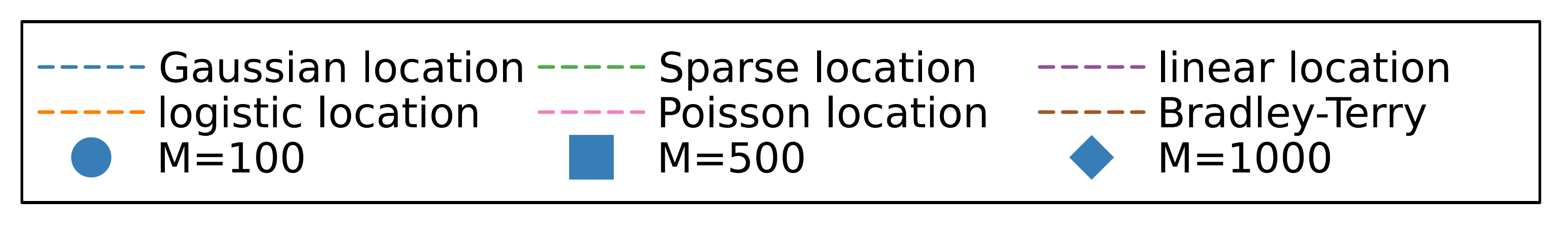}
    \caption{Relative Coreset MCMC posterior approximation error comparing ADAM (with different learning rates) and 
    the proposed Hot DoG method (under our recommended setting).
    The metric plotted is
    the ratio of average squared z-scores (defined in \cref{eq:avg_sq_z}) under ADAM to those under Hot DoG.
    Values above the horizontal black line ($10^0$) indicate that the proposed Hot DoG method outperformed ADAM.
    Median values after 200,000 optimization iterations across 10 trials
    are used for the relative comparison for a variety of datasets, models, and coreset sizes.
    }
    \label{fig:NVDoG_ADAM}
\end{figure}

Bayesian inference provides a flexible framework for parameter estimation and uncertainty quantification in 
statistical models. Markov chain Monte Carlo [\citealp{robert1999monte}; \citealp{robert2011short}; 
\citealp[Chs.~11 and 12]{gelman2013bayesian}], the standard methodology for 
performing Bayesian inference, involves simulating carefully constructed Markov chains whose stationary distribution 
is the target Bayesian posterior. In the large-scale data setting, this procedure can become prohibitively expensive, 
as it requires iterating over the entire data set to simulate the next state. 

\emph{Bayesian coresets} \citep{huggins2016coresets}
are a popular approach for speeding up Bayesian inference in the 
large-scale data setting.
A Bayesian coreset
is a weighted subset of data that replaces the full data set 
during inference, leveraging the insight that large datasets often exhibit a 
significant degree of redundancy.\footnote{A related approach, \emph{data distillation}, constructs a small 
synthetic data set for downstream tasks. However, this approach often requires bespoke methods for non-real-valued data
(see [\citealp[Sec.~3]{sachdeva2023data}]). In contrast, Bayesian coresets do not modify individual data points, 
and so are fully generic.} 
With a carefully constructed coreset, one can significantly reduce the computational cost 
of inference while still obtaining samples from a high quality 
approximation of the full Bayesian posterior. In fact, given a data set of $N$ points, a 
coreset of size $\scO\left(\log N\right)$ is sufficient for providing a near-exact posterior approximation 
in exponential family and other sufficiently simple models [\citealp[Thms.~4.1 and 4.2]{naik2022fast}; \citealp[Prop.~3.1]{chen2022bayesian}]
and $\scO\lt(\operatorname{polylog} N\rt)$ is sufficient for more general cases \citep[Cor.~6.1]{campbell2024general}.

Constructing a coreset involves picking the data points to include in the coreset and assigning each data point its 
corresponding weight. The state-of-the-art method, Coreset MCMC \citep{chen2024coreset}, selects coreset 
points by sampling them uniformly from the full data set, and learns the weights using stochastic gradient optimization techniques, e.g., ADAM \citep{kingma2014adam}, 
where the gradients are estimated using MCMC draws targeting the current coreset posterior. 
However, as we demonstrate in this paper, there are two issues with this approach.
First, the quality of the constructed coreset is sensitive to the learning rate of the 
stochastic optimization algorithm. And second, gradient estimates using MCMC draws
are affected strongly in early iterations by initialization bias, leading to poor 
optimization performance.

To address these challenges, we first propose 
\emph{Hot-start Distance over Gradient} (Hot DoG), a tuning-free stochastic
gradient optimization procedure that can be used for learning coreset weights
in Coreset MCMC. Hot DoG is a stochastic gradient method combining techniques from Do(W)G
\citep{ivgi2023dog,khaled2023dowg}, ADAM \citep{kingma2014adam}, and RMSProp
\citep{hinton2012neural} to set learning rates automatically. Hot DoG also includes an
automated warm-up phase prior to weight optimization, which guards against usage
of low quality MCMC draws when estimating the objective function gradients.
We then analyze the convergence behaviour of Hot DoG in a representative setting.
Empirically, \cref{fig:NVDoG_ADAM} demonstrates that Hot DoG under our recommended setting
performs competitively to optimally-tuned ADAM across a wide range of models, datasets, and coreset sizes, 
and can be multiple orders of magnitude more accurate than ADAM using other learning rates.
Beyond the results shown in \cref{fig:NVDoG_ADAM}, we provide an extensive 
empirical investigation of the reliability of Hot DoG in comparison to other methods across 
various synthetic and real experiments.

\section{Background}
\label{sec:background}

\subsection{Bayesian Coresets}
We are given a data set $(X_n)_{n=1}^N$ of $N$ observations, 
a log-likelihood $\ell_n \coloneqq \log p(x_n \mid \theta)$ 
for observation $n$ given $\theta \in \Theta$, and a prior density $\pi_0(\theta)$. 
We would like to sample from the  Bayesian posterior with density
\[
  \pi(\theta) \coloneqq \frac{1}{Z} \exp\left( \sum_{n=1}^N \ell_n(\theta) \right) \pi_0(\theta),
\]
where $Z$ is the unknown normalizing constant. A Bayesian coreset replaces the sum over $N$ log-likelihood terms with a 
weighted sum over a subset of size $M$, where $M\ll N$. Without loss of generality, we assume that these are the 
first $M$ points. The coreset posterior can then be written as
\[
  \pi_w(\theta) \coloneqq \frac{1}{Z(w)} \exp\left( \sum_{m=1}^M w_m \ell_m(\theta) \right) \pi_0(\theta), 
  \label{eq:coresetposterior}
\]
where $w \in \reals^M_{+}$ is a vector of coreset weights.
Recent coreset construction methods 
uniformly select $M$ points to include in the coreset \citep{naik2022fast,chen2022bayesian,chen2024coreset}, and 
then optimize the weights of those $M$ points as 
a variational inference problem \citep{campbell2019sparse},
\[
    w^\star = \argmin_{w\in\reals^M} \kl{\pi_w}{\pi} \quad \text{s.t.} \quad w \in \mathcal{W}\label{eq:coresetopt},
\]
with objective function gradient
\[
    \label{eq:grad}
    &\nabla_w \kl{\pi_w}{\pi} \\
    = &\Cov_{\pi_w}\lt( \bbmat \ell_1(\theta) \\ \vdots \\ \ell_M(\theta)\ebmat, 
        \sum_m w_m\ell_m(\theta) - \sum_n \ell_n(\theta) \rt).
\] 

\balg[t]
\caption{\texttt{CoresetMCMC}} \label{alg:coresetmcmc}
\balgc
\Require $\theta_0$, $\kappa_w$, $S$, $M$
\LineComment Initialize coreset weights
\State $w_{0m} = \frac{N}{M}, \quad m = 1,\cdots,M$
\For{$t=0, \dots, T$}
    \LineComment Subsample the data
    \State $\scS_{t} \gets \Unif\lt(S, [N]\rt)$ (without replacement)
    \LineComment Compute gradient estimate
    \State $\hat{g}_{t} \gets g(w_{t}, \theta_{t}, \scS_{t})$ (\cref{eq:gradest})
    \State $w_{t+1} \gets $ \texttt{stochastic\_gradient\_step($w_{t}, \hat{g}_{t}$)}
    \LineComment Step each Markov chain
    \For{$k=1, \dots, K$}
        \State $\theta_{(t+1)k} \dist \kappa_{w_{t+1}}(\cdot \mid \theta_{tk})$
    \EndFor
\EndFor
\ealgc
\ealg

\subsection{Coreset MCMC}
The key challenge in solving \cref{eq:coresetopt} is that $\pi_w$ does not admit tractable \iid draws,
and so unbiased estimates of the gradient in \cref{eq:grad} are not readily available.
Coreset MCMC \citep{chen2024coreset} is an adaptive algorithm that addresses this issue.
The method first initializes weights $w_0 \in\reals^M$ and 
$K\geq 2$ samples $\theta_0 = \lt(\theta_{01}, \dots, \theta_{0K}\rt) \in \Theta^K$.
At iteration $t\in\nats$, given coreset weights $w_t$ and samples 
$\theta_t \in \Theta^K$,
it then updates the weights 
$w_t \to w_{t+1}$ using the stochastic gradient estimate based on the draws $\theta_t$,
\[
    \label{eq:gradest}
    &g(w_t, \theta_t, \scS_t) = \\
    &\frac{1}{K\!-\!1}\!\sum_{k=1}^K \!\!\bbmat \bar\ell_1(\theta_{tk})\\ \vdots \\ \bar\ell_M(\theta_{tk})\ebmat
    \!\!\lt(\!\sum_m w_{tm}\bar\ell_m(\theta_{tk}) \!-\! \frac{N}{S}\!\sum_{s\in \scS_t}\!\bar\ell_{s}(\theta_{tk}) \!\rt),
\]
where $\scS_t \subseteq [N]$ is a uniform subsample of indices of size $S$, 
and $\bar\ell_n(\theta_{tk}) = \ell_n(\theta_{tk}) - \frac{1}{K}\sum_{j=1}^K \ell_n(\theta_{tj})$.
To complete the iteration, the method updates the samples by independently drawing  $\theta_{(t+1)k} \dist \kappa_{w_{t+1}}(\theta_{tk}, \cdot)$ for each $k\in [K]$,
where $\kappa_w$ is a family of $\pi_w$-invariant Markov kernels. %with invariant distribution $\pi_w$.
The pseudocode for Coreset MCMC is outlined in \cref{alg:coresetmcmc}.

\begin{figure}[t]
    \includegraphics[width=\columnwidth]{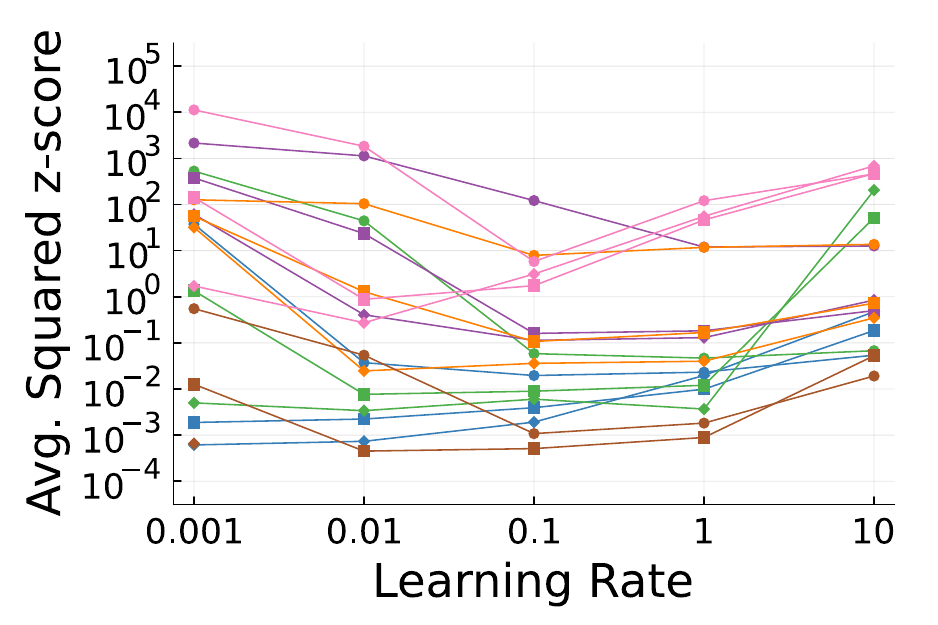}
    \includegraphics[width=\columnwidth]{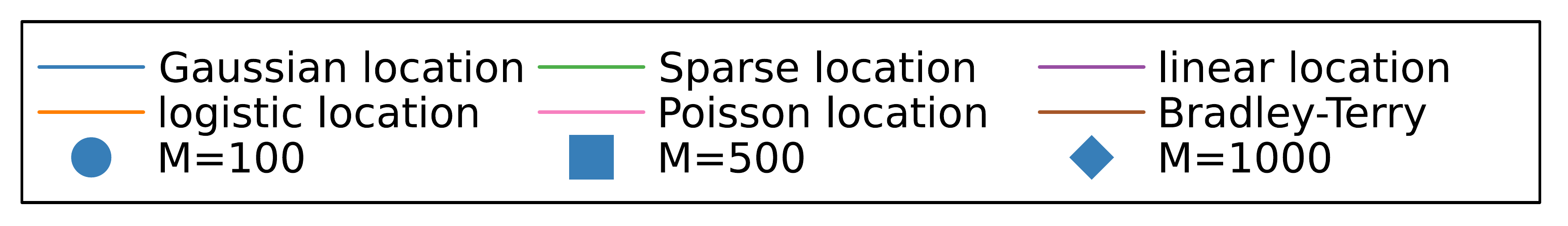}
    \caption{Coreset MCMC posterior approximation error (as defined in \cref{eq:avg_sq_z})
    		using ADAM with different learning rates for a variety of datasets, models, and coreset sizes.
    		The lines indicate median values after 200,000 optimization iterations across 10 trials.}
    \label{fig:ADAM}
\end{figure}

\begin{figure*}[t!]
    \centering{
    \begin{subfigure}{0.45\textwidth}
        \includegraphics[width=\columnwidth]{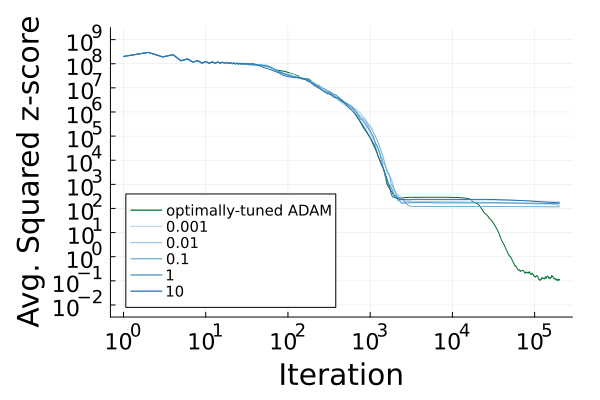}
        \caption{DoG}\label{fig:DoG_nomix}
    \end{subfigure}
    \begin{subfigure}{0.45\textwidth}
        \includegraphics[width=\columnwidth]{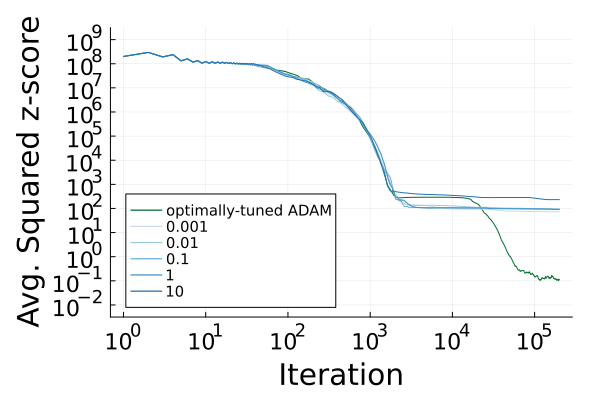}
        \caption{DoWG}\label{fig:DoWG_nomix}
    \end{subfigure}}\\
    \centering{
    \begin{subfigure}{0.45\textwidth}
        \includegraphics[width=\columnwidth]{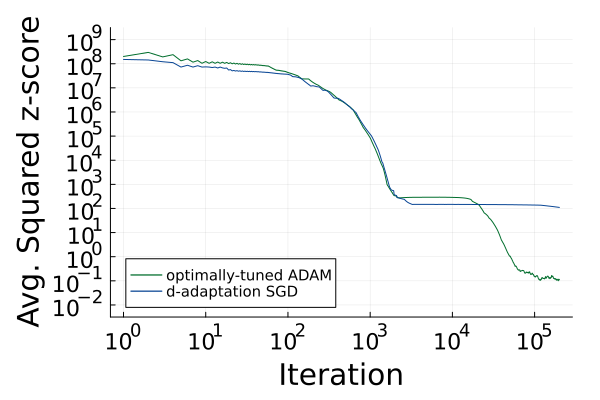}
        \caption{D-Adaptation SGD}\label{fig:dadaptSGD_nomix}
    \end{subfigure}
    \begin{subfigure}{0.45\textwidth}
        \includegraphics[width=\columnwidth]{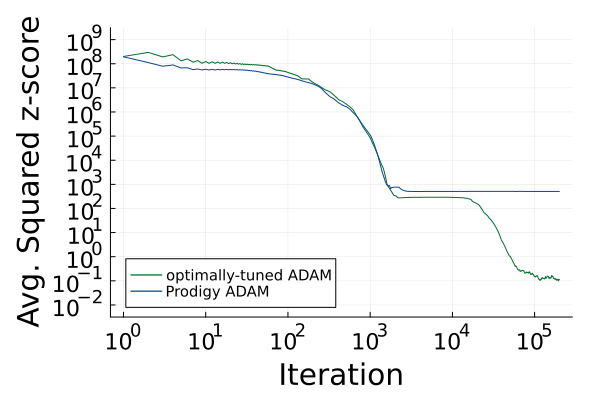}
        \caption{prodigy ADAM}\label{fig:ProdigyADAM_nomix}
    \end{subfigure}}
    \caption{Traces of average squared coordinate-wise z-scores (defined in \cref{eq:avg_sq_z}) between the true and approximated posterior for 
                a Bayesian linear regression example with $M=1{,}000$ coreset points. 
                We evaluate four learning-rate-free SGD methods: 
                DoG and DoWG (with varying initial learning rate parameter), and 
                D-Adaptation SGD and prodigy ADAM (with default initial lower bound $10^{-6}$).
                The optimally-tuned ADAM baseline is shown in green. Results display the median 
                after 200,000 optimization iterations across 10 trials.}
    \label{fig:direct_application}
\end{figure*}

\section{Tuning-Free Coreset MCMC}

A key design choice when using Coreset MCMC is to specify how gradient estimates are used to
optimize the weights. One can use ADAM \citep{kingma2014adam}, which is
used as the default optimizer for Coreset MCMC \citep{chen2024coreset}:
at iteration $t$, with $\gamma_t > 0$ being the user-specified learning rate, we set 
\[
    w_{t+1} \gets \proj_{\geq 0}\lt(w_t - \gamma_t \frac{\hat{m}_t}{\sqrt{\hat{v}_t} + \epsilon}\rt),
\]
where $\hat{m}_t$ and $\hat{v}_t$ are exponential averages of past gradients $(\hat{g}_i)_{i=0}^t$ and their 
element-wise squares, and $\epsilon$ is a small constant.
There are a wide range of other first-order stochastic methods available that could be used (e.g., vanilla stochastic
gradient descent, AdaGrad \citep{duchi2011adaptive}, etc.). However,
like ADAM, most of these algorithms require setting a learning rate $\gamma_t$. And as we show in \cref{fig:ADAM}, the quality 
of samples obtained from Coreset MCMC can be highly sensitive to the selected learning rate. In 
particular, \cref{fig:ADAM} shows that when using ADAM, no 
single learning rate applies well across all problems and coreset sizes; and for a given problem, the performance can vary
by orders of magnitude as one varies the learning rate.
Furthermore, the default ADAM learning rate of $10^{-3}$ \citep{kingma2014adam} 
provides poor results in most of the problems tested. As a result, 
careful tuning of the learning rate is required to obtain high quality posterior approximations.
This usually involves a search on a log-scaled grid, which is computationally wasteful as the results 
for all but one of the parameter values are thrown out. Moreover, in practice determining which learning rate
provides the best posterior approximation may not be straightforward, as we do not have access to estimates of the objective function.

A number of recent works in the literature propose learning-rate-free stochastic optimization methods to address this issue
\citep{carmon2022making,ivgi2023dog,khaled2023dowg,defazio2023learning,mishchenko2023prodigy}. Many of these methods 
are shown empirically to work competitively compared to optimally-tuned SGD on a wide range of large-scale, 
non-convex deep learning problems. Although different at first glance, all of these methods 
arise from the same insight. Suppose one would like to solve the stochastic 
optimization problem
\[
    \min_{w\in\reals^d} \E\left[ f(w,\theta) \right],
\]
where for all $\theta$, $f(\cdot, \theta)$ is convex and we only have access to unbiased stochastic gradient $g_t = \partial f(w_t, \theta_t)$. 
Define the initial distance to the optimal solution $d_0 = \|w_0 - w^\star\|$ 
and the sum of all gradient norms $G_T = \sum_{t\leq T}\|g_t\|^2$.
By setting the SGD learning rate
    $\gamma^\star = \frac{d_0}{\sqrt{G_T}}$,
the average iterate $\bar{w} = \frac{1}{T}\sum_{t\leq T}w_t$ satisfies the optimal error bound 
\[
    \E\left[ f(\bar{w},\theta) \right] - \E\left[ f(w^\star,\theta) \right] \leq \frac{d_0\sqrt{G_T}}{T}
\]
after $T$ iterations \citep{carmon2022making,orabona2020icml}.
Learning-rate-free methods therefore essentially try to estimate or bound the initial distance to the optimal solution $d_0$,
which is unknown in practice. To the best of our knowledge, there are four state-of-the-art methods
that do this in a manner that does not require multiple optimization runs, knowledge of unknown constants,
or the ability to query the objective function:
DoG \citep{ivgi2023dog}, DoWG \citep{khaled2023dowg}, D-Adaptation \citep{defazio2023learning} and prodigy 
\citep{mishchenko2023prodigy}. 
DoG and DoWG run vanilla stochastic gradient descent (SGD),
\[
    w_{t+1} &\gets \proj_{\geq 0}\lt(w_t - \gamma_t g_t\rt),
\]
with learning rate schedules
\[
    \hspace{-.3cm}\gamma_t = \frac{r_t}{\sqrt{G_t}} \text{(DoG)},\,
    \gamma_t = \frac{r^2_t}{\sqrt{\sum_{i\leq t} r_{i}^2 \|g_{i}\|^2}} \text{(DoWG)}, \label{eq:learningrates}
\]
where $r_0$ is set to some small constant
and, for $t \ge 1$, 
\[
r_t = \max_{i\leq t} \|w_t - w_0\|.
\]
For D-Adaptation and prodigy, 
$r_t$ in \cref{eq:learningrates} is replaced with a lower bound $d_t$ on $d_0$,
which is updated using estimated correlations between the $g_t$ and 
step direction $w_0-w_t$:
\[
    d_{t+1} = \max\left\{ \frac{\sum_{i=0}^{t} d_i \left\langle g_i, w_0-w_i \right\rangle}
                                { \left\| \sum_{i=0}^t d_ig_i \right\|}, d_{t} \right\}.
\]
D-Adaptation replaces $r_t$ in \cref{eq:learningrates} (DoG) with $d_t$, while prodigy replaces $r_t$ in \cref{eq:learningrates} (DoWG)
with $d_t$. Both D-Adaptation and prodigy have SGD and ADAM-based variants.
All four methods have been shown empirically to match the performance of 
optimally-tuned SGD.

\cref{fig:direct_application} shows the results from direct applications of DoG, DoWG, D-Adaptation (SGD), and prodigy 
(ADAM) to Coreset MCMC. We see that the quality of posterior approximation from all of four methods are orders of 
magnitude worse than optimally-tuned ADAM. With $\theta_0$ initialized far away from high density regions of 
$\pi_{w_0}$, the initial gradient estimates are large in magnitude, which leads to small learning rates. The 
accumulation of these large gradient norms in the learning rate denominator eventually causes the learning rate to 
vanish, halting the progress of coreset weight optimization. We address these problems in the 
next section.

Before concluding this section, we note that there are other approaches for making SGD free of learning rate 
tuning: some methods involve using stochastic versions of line search 
\citep{vaswani2019painless,paquette2020stochastic}, and others do the same for the Polyak step size 
\citep{loizou2021stochastic}. These methods are not applicable in our setting as they require evaluating the 
objective function. Recall that due to the unknown $Z(w)$ term in \cref{eq:coresetposterior}, we do not have access 
to estimates of the objective function.

\balg[t]
\caption{\texttt{HotDoG}} \label{alg:NVDoG}
\balgc

\Require $\beta_1 = 0.9$, $\beta_2 = 0.999$, $\epsilon = 10^{-8}$, $r = 10^{-3}$\\
$\quad\quad\quad\quad T$, $\theta_0$, $w_0$
\State $v_0 \gets \bm{0}$, $m_0 \gets \bm{0}$, $d_0 \gets \bm{0}$, $c \gets 0$, $h \gets \texttt{false}$
\For{$t=1, \dots, T$}
    \If{h} 
        \State $c\gets c+1$
        \State $\scS_{t} \gets \Unif\lt(S, [N]\rt)$ (without replacement)
        \State $\hat{g}_t = g(w_{t-1}, \theta_{t-1}, \scS_t)$ (\cref{eq:gradest})
        
        \State $v_t \gets \beta_2 v_{t-1} + (1-\beta_2) \hat{g}_t^2$
        \State $m_t \gets \beta_1 m_{t-1} + (1-\beta_1) \hat{g}_t$
        \State $d_t \!\gets\! \beta_1 d_{t-1} \!+\! (1\!-\!\beta_1) \max\left\{ \left| w_{t-1} \!-\! w_0 \right|, d_{t-1} \right\}$ 
        \State $\hat{v}_t \gets v_t / (1 - \beta_2^c)$
        \State $\hat{m}_t \gets m_t / (1 - \beta_1^c)$
        \State $\hat{d_t} \gets $ ( $r\mathbf{1}$ \algorithmicif\ {t==1} \algorithmicelse\ $d_t / (1 - \beta_1^{c-1})$ )
        \State $w_t \gets w_{t-1} \!-\! \hat{d}_t \left(\diag\left(\left(c \left(\hat{v}_t + \epsilon\right)\right)^{\frac{1}{2}}\right)\right)^{-1}\! \odot \hat{m}_t$ %\Comment{update step}
    \Else
        \State $w_t \!\gets\! w_{t-1}$, $v_t \!\gets\! v_{t-1}$, $m_t \!\gets\! m_{t-1}$, $d_t \!\gets\! d_{t-1}$
    \EndIf
    \For{$k=1, \dots, K$}
        \State $\theta_{tk} \dist \kappa_{w_{t}}(\cdot \mid \theta_{(t-1)k})$ \Comment{record $\ell_{tk}$}
    \EndFor
    \LineComment Hot-start test
    \State $h \!\gets\! $ (true \algorithmicif\ $h$ \algorithmicelse\ $\texttt{HotStartTest}\!\left(\!\left(\!\ell_{ik}\!\right)_{i=1,k=1}^{t,K}\!,t\!\right)$)
\EndFor

\State\Return $w_T$

\ealgc
\ealg

\section{Hot DoG}
\label{sec:nvdog}
In this section, we develop our novel Markovian optimization method, \emph{Hot-start DoG} (Hot DoG),
presented in \cref{alg:NVDoG}. Our method extends the original DoG optimizer in two ways: 
(1) we add a tuning-free hot-start test that automatically detects when the 
Markov chains have properly mixed and stochastic gradient estimates are stable, 
at which point we start coreset weight optimization;
and (2) we apply acceleration techniques to DoG.

\subsection{Hot-Start Test}\label{sec:hotstarttest}
Poorly initialized Markov chain states $\theta_0$ can be detrimental to 
the performance of learning-rate-free methods in Coreset MCMC. 
\cref{fig:burnintest}, especially \cref{fig:burnintest-linear,fig:burnintest-logistic,fig:burnintest-poiss}
show that this is likely due to the bias of initial gradient estimates. %In particular,
When $\theta_0$ is initialized far away from high density regions of $\pi_{w_0}$, the initial gradient estimates 
can have norms that are orders of magnitude larger than those 
computed using \iid draws. This leads to a quickly vanishing learning rate in \cref{eq:learningrates}. 
Therefore, it is crucial to 
hot-start the Markov chains to ensure they are properly mixed before training the coreset weights.
There are MCMC convergence diagnostics
that could be used for this purpose (e.g, $\shR$ \citep{vehtari2021rank}); many 
work only with real-valued variables, and are overly stringent for our application.
We require a test that works for general coreset posteriors of the form \cref{eq:coresetposterior}
and checks only that gradient estimates have stabilized reasonably.

To address this challenge, we propose keeping the weights fixed at their initialization
(i.e., $w_{t+1} \gets w_t$) until a hot-start test passes.
For the test, for each Markov chain $k\in [K]$, 
we split the iterates $i=1,\dots, t$ into 3 segments, each
of equal length $n = \ceil{t/3}$. 
We compute the average log-potentials for the two latter segments $m_{k1}$, $m_{k2}$,
and the standard deviations of residual errors $s_{k1}, s_{k2}$ from a linear fit 
\[
m_{ki} \!=\! \frac{\sum_{j=in\!+\!1}^{(i\!+\!1)n}\! \ell_{jk}}{n},
s_{ki}^2 \!=\! \frac{\min_{\substack{a,b}}\! \sum_{j=in\!+\!1}^{(i\!+\!1)n} (a\! +\! bj\! -\! \ell_{jk})^2}{n-2}.
\]
Here $\ell_{jk} = \sum_{m'=1}^M w_{0m'}\ell_{m'}(\theta_{jk})$ is the log-potential for chain $k$ at iteration $j$.
Our test monitors the difference between $m_{k1}$ and $m_{k2}$ relative to $s_{k1}$ and $s_{k2}$. 
A small difference in the averages indicates that the chains have stabilized. 
The residual standard errors allows us to remove trends from the noise computation.
We define, for each $k\in[K]$,
\[
    u_k = \frac{\lt| m_{k1} - m_{k2} \rt|}{\max\{s_{k1}, s_{k2}\}},
\]
and use the median of $\left( u_{k} \right)_{k=1}^K$ as our test statistic. This test statistic is checked against 
a threshold $c$; the test passes when the median test statistic is less than $c$. 
\cref{alg:burnintermination} shows the pseudocode for the hot-start test.
We find in practice setting $c=0.5$ works well in general.

\balg[t]
\caption{\texttt{HotStartTest}} \label{alg:burnintermination}
\balgc
\Require $\left(\ell_{ik}\right)_{i=1, k=1}^{t, K}$, $t$, $c=0.5$
\State $n = \texttt{ceil}(t/3)$
\For{$k=1, \dots, K$}
    \State $s^2_{k1} \gets \frac{1}{n-2}\min_{a,b\in\reals}\sum_{i=n+1}^{2n} \left( a+b i - \ell_{ik} \right)^2$
    \State $s^2_{k2} \gets \frac{1}{n-2}\min_{a,b\in\reals}\sum_{i=2n+1}^{t} \left( a+b i - \ell_{ik} \right)^2$
    \State $u_{k} \gets \frac{\left|\left(\frac{1}{n}\sum_{i=n+1}^{2n} \ell_{ik}\right) - \left(\frac{1}{n}\sum_{i=2n+1}^{t}\ell_{ik}\right)\right|}{\max\{s_{k1}, s_{k2}\}}$
\EndFor
\State \Return (true \algorithmicif\ \texttt{median}$\left( u_{1}, \dots, u_{K} \right) < c$ \algorithmicelse\  false)
\ealgc
\ealg

\subsection{Acceleration}
To accelerate DoG, we begin by noting that 
the denominator of the DoG learning rate in \cref{eq:learningrates} is similar to that of 
AdaGrad \citep{duchi2011adaptive} in that it is a cumulative sum of some function of the gradient. Therefore, 
we can leverage the idea used in RMSProp \citep{hinton2012neural} for accelerating AdaGrad to accelerate DoG. 
In particular, at iteration $t$, we can replace $\sum_{i\leq t} \|\hat{g}_i\|^2$ with 
$t\hat{v}_t$, the bias-corrected exponential moving average of the squared gradient. 
This allows us to exponentially decrease the weights of past gradient norms. As 
a result, the effect of the early $\|\hat{g}_t\|^2$ terms on the learning rate diminishes over time, 
resulting in less conservative learning rates. 
To account for situations where the gradient estimates differ in scale across dimensions, we apply 
the above acceleration technique in a coordinate-wise fashion and obtain the following update rule for $\hat{v}_t$:
\[
    v_t = \beta_2 v_{t-1} + (1-\beta_2) \hat{g}_t^2, \quad \hat{v}_t = \frac{v_t}{1-\beta_2^t},
\]
where $\beta_2 \in (0,1)$ is the exponential decay rate, $v_0 = 0$, and $\hat{g}_t^2$ denotes the vector with each entry of $\hat{g}_t$ squared.
We further apply the same idea to $r_t$, the maximum distance traveled from $w_0$, and $\hat{g}_t$, the gradient estimate itself. 
We use $\beta_1 \in (0,1)$ to denote the exponential decay rate for these two quantities.
Our final proposed optimization procedure is outlined in \cref{alg:NVDoG}.
Note that in \cref{alg:NVDoG}, all computations are coordinate-wise.

In Hot DoG, we set the exponential decay rates, $\beta_1$ and $\beta_2$, to be the same as those in 
\citet{kingma2014adam}, and we set the initial learning rate $r$ to a small constant (default $10^{-3}$) following the 
recommendation of \citet{ivgi2023dog}. 

\subsection{Convergence Analysis}
In this subsection, we present a theoretical analysis of the convergence of the coreset weights produced by Hot DoG. 
We begin by stating the set of assumptions, under which our analysis is conducted.
These assumptions are stated formally stated in \cref{sec:assumptions}.
As required by \cref{alg:NVDoG}, we have that $|\beta_1|<1$, $|\beta_2|<1$, and $\epsilon, r > 0$. 
We further impose a set of assumptions about the feasible region $\mathcal{W}$ of the coreset weights.
Namely, we assume (1) the coreset weights are non-negative and their sum is bounded 
above by a constant $B$ (\cref{assump:exact,assump:constraint}), and (2) the existence of an exact coreset 
$w^\star\in\mathcal{W}$ in the sense that $\kl{\pi_{w^\star}}{\pi}=0$.
Both of these assumptions greatly simplify the analysis without sacrificing the representative nature of our assumed model.
A typical choice for the coreset weight bound is to set $B=N$, where $N$ is the total number of observations. 
In terms of the optimal coreset, past work has shown that it provides a near-exact approximation with high probability 
for the wide class of strongly log-concave models [\citealp[Thm.~4.2]{naik2022fast}].
Under \cref{assump:exact}, which is similar to Assumption 3.1 in \citet{chen2024coreset},
we do not expect the convergence result to change in a meaningful way, aside from there being a persistent error 
corresponding to the optimal coreset error. 

Finally, we state our assumptions regarding the stochastic gradient (\cref{eq:gradest}), which estimates \cref{eq:grad}. 
We assume that the stochastic gradients are uniformly bounded above by a constant $U$ (\cref{assump:grad_bound}).
Now note that in \cref{eq:gradest}, Monte Carlo error from the MCMC samples $\theta_t$ contributes to the stochasticity. 
We additionally impose a mixing condition on the Markov chains (\cref{assump:mixing}), 
and assume that the Monte Carlo error is controlled (\cref{assump:noise}).

We now present our main theorem in \cref{thm:convergence}. 
The proof of \cref{thm:convergence} can be found in \cref{sec:proof}.
Our result shows that Hot DoG produces coreset weights that converge to the optimum in expectation at 
a sublinear rate.
This convergence rate is consistent with ADAM and other 
learning-rate-free stochastic gradient methods discussed in the paper 
(see for example [\citealp[Thm.~3.10]{ivgi2023dog} and [\citealp[Thm.~2]{mishchenko2023prodigy}]).

\begin{theorem}[Hot DoG convergence]\label{thm:convergence}
    Suppose \cref{assump:constraint,assump:exact,assump:grad_bound,assump:mixing,assump:noise} hold.
    As $t \to \infty$,
    \[
        \E\|w_t \!-\! w^\star\|^2 = O\left( \frac{1}{\sqrt{t}} \right).
    \]
\end{theorem}

It is worth noting that whether to employ the hot-start test does not alter the convergence rate of Hot DoG as shown 
in \cref{thm:convergence}. Instead, the hot-start test can lead to a more favourable constant in the 
convergence rate.
As we discussed in \cref{sec:hotstarttest}, the hot-start test helps avoid updating the coreset weights using 
initial gradient estimates that may have unusually large norms. 
In terms of our analysis, by holding off optimizing $w$ until the hot-start test passes, 
we can obtain a tighter bound on the gradient norm (i.e., a smaller $U$ in \cref{assump:grad_bound}). 
This results in a smaller constant in the convergence rate given in \cref{thm:convergence}, 
ultimately leading to improved finite-time performance.

\section{Experiments}
\label{sec:experiments}

\begin{figure*}
    \begin{subfigure}{0.33\textwidth}
    \includegraphics[width=\columnwidth]{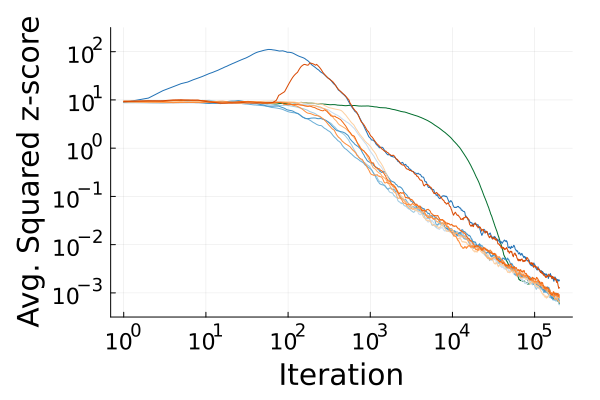}
    \caption{Gaussian location}
    \end{subfigure}
    \begin{subfigure}{0.33\textwidth}
    \includegraphics[width=\columnwidth]{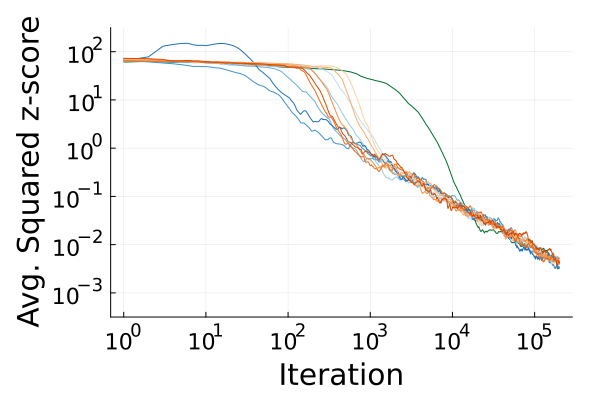}
    \caption{Sparse regression}
    \end{subfigure}
    \begin{subfigure}{0.33\textwidth}
    \includegraphics[width=\columnwidth]{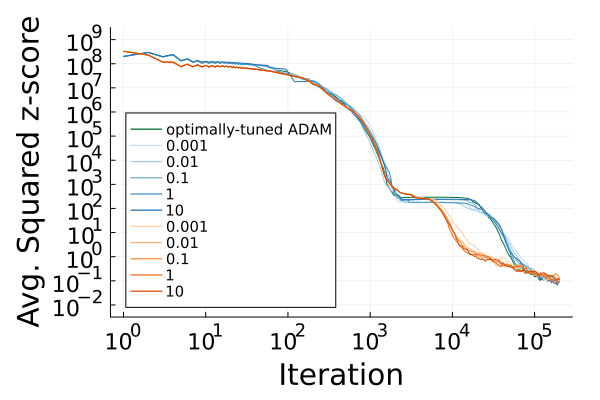}
    \caption{Linear regression}\label{fig:burnincomparison-linear}
    \end{subfigure}
    \begin{subfigure}{0.33\textwidth}
    \includegraphics[width=\columnwidth]{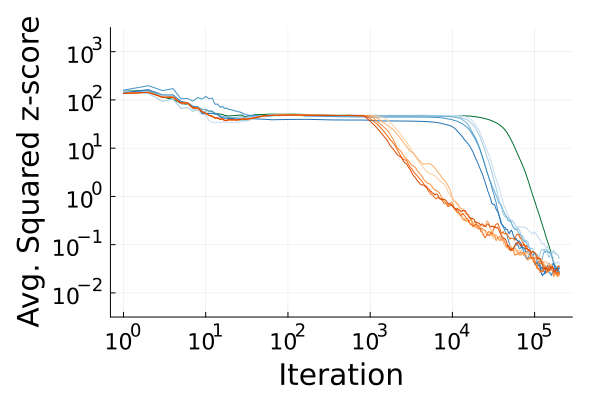}
    \caption{Logistic regression}
    \end{subfigure}
    \begin{subfigure}{0.33\textwidth}
    \includegraphics[width=\columnwidth]{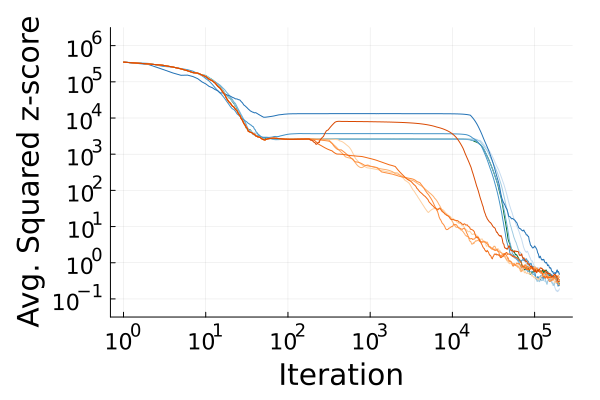}
    \caption{Poisson regression}
    \end{subfigure}
    \begin{subfigure}{0.33\textwidth}
    \includegraphics[width=\columnwidth]{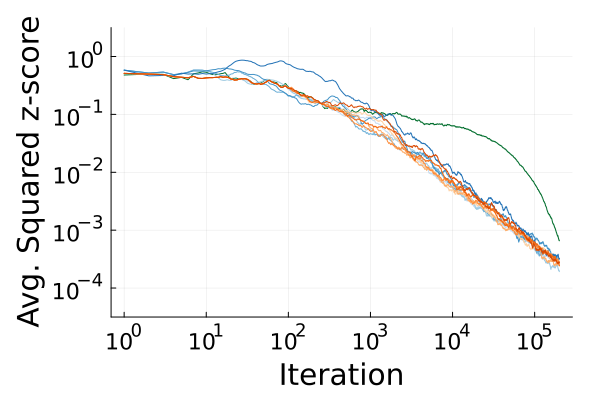}
    \caption{Bradley-Terry}
    \end{subfigure}
    \caption{Traces of average squared coordinate-wise z-scores between the true and approximated posterior 
    across all experiments, obtained using Hot DoG with and without hot-start test. 
    All figures share the legend in \cref{fig:burnincomparison-linear}. The coreset size $M$ is $1000$ and each line
    represents a different initial learning rate parameter. The lines indicate the median from $10$ runs.
    Orange lines indicate runs with hot-start test and blue lines without.}
    \label{fig:burnincomparison}
\end{figure*}

\begin{figure*}
    \begin{subfigure}{0.33\textwidth}
        \includegraphics[width=\columnwidth]{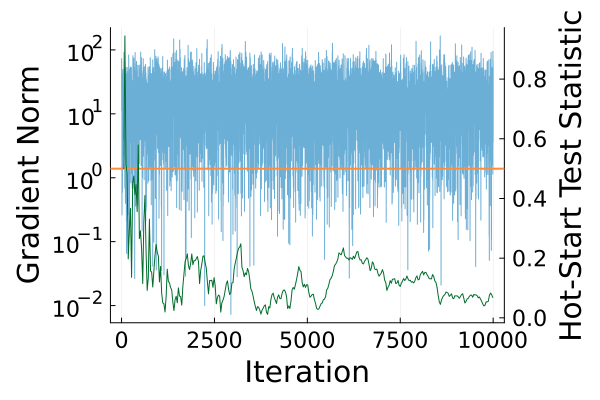}
        \caption{Gaussian location}
    \end{subfigure}
    \begin{subfigure}{0.33\textwidth}
        \includegraphics[width=\columnwidth]{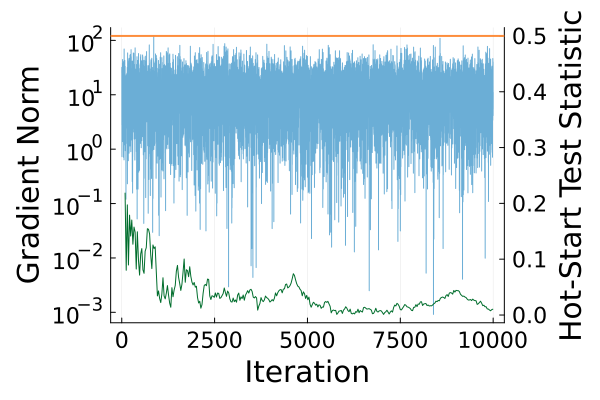}
        \caption{Sparse regression}
    \end{subfigure}
    \begin{subfigure}{0.33\textwidth}
        \includegraphics[width=\columnwidth]{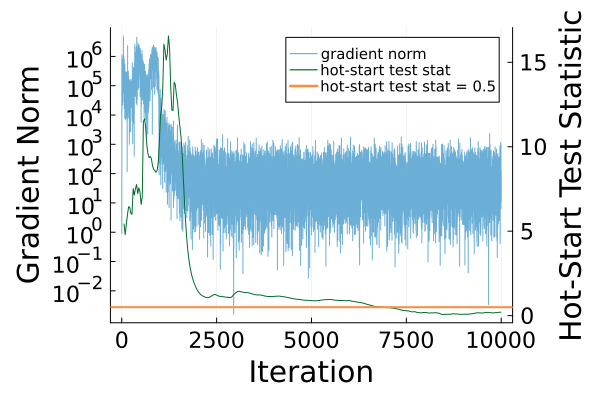}
        \caption{Linear regression}\label{fig:burnintest-linear}
    \end{subfigure}
    \begin{subfigure}{0.33\textwidth}
        \includegraphics[width=\columnwidth]{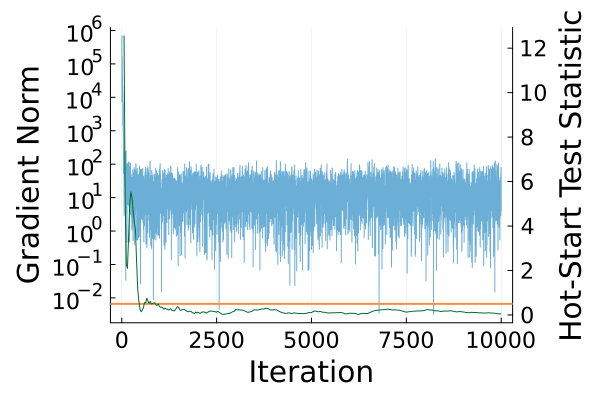}
        \caption{Logistic regression}\label{fig:burnintest-logistic}
    \end{subfigure}
    \begin{subfigure}{0.33\textwidth}
        \includegraphics[width=\columnwidth]{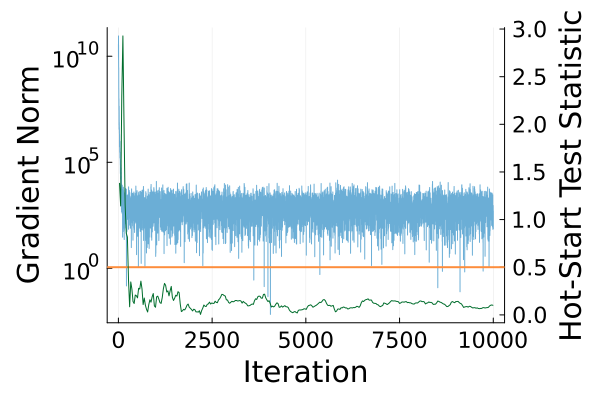}
        \caption{Poisson regression}\label{fig:burnintest-poiss}
    \end{subfigure}
    \begin{subfigure}{0.33\textwidth}
        \includegraphics[width=\columnwidth]{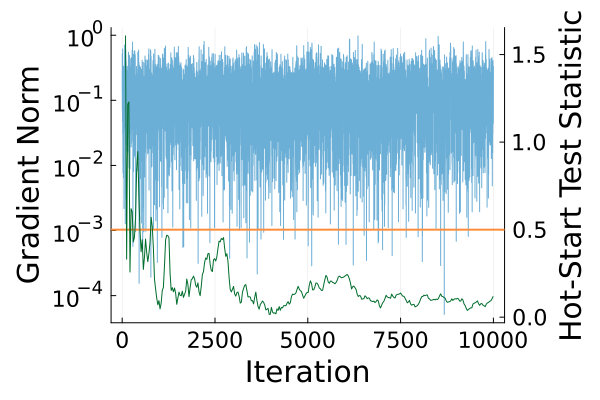}
        \caption{Bradley-Terry}
    \end{subfigure}
    \caption{Trace of gradient estimate norms (blue) and hot-start test statistics (green) before weight optimization
            across all experiments with $M=1000$.
            The orange horizontal line is the test statistic threshold $c=0.5$.}
    \label{fig:burnintest}
\end{figure*}

\begin{figure*}
    \begin{subfigure}{0.33\textwidth}
    \includegraphics[width=\columnwidth]{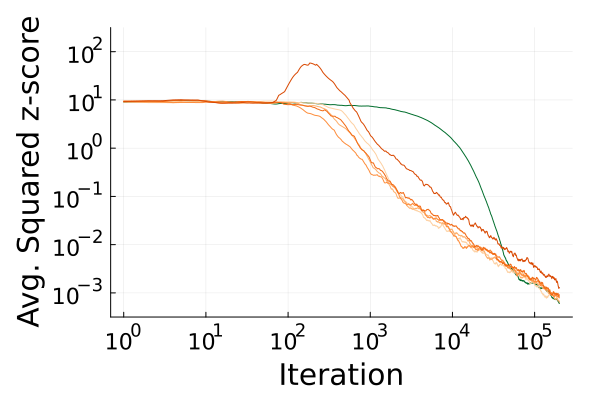}
    \caption{Gaussian location}
    \end{subfigure}
    \begin{subfigure}{0.33\textwidth}
    \includegraphics[width=\columnwidth]{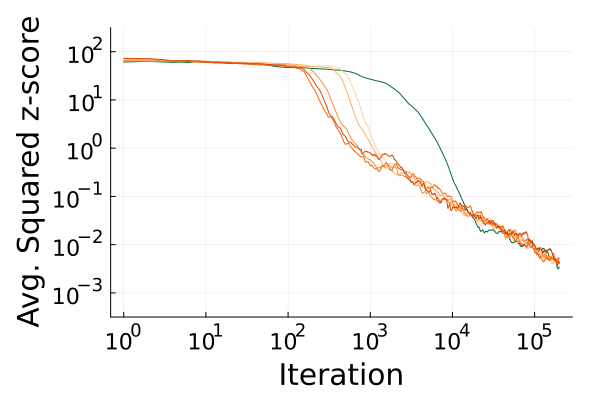}
    \caption{Sparse regression}
    \end{subfigure}
    \begin{subfigure}{0.33\textwidth}
    \includegraphics[width=\columnwidth]{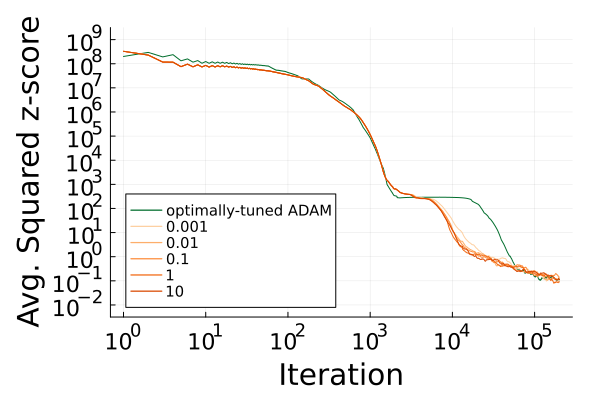}
    \caption{Linear regression}\label{fig:tracecombined-linear}
    \end{subfigure}
    \begin{subfigure}{0.33\textwidth}
    \includegraphics[width=\columnwidth]{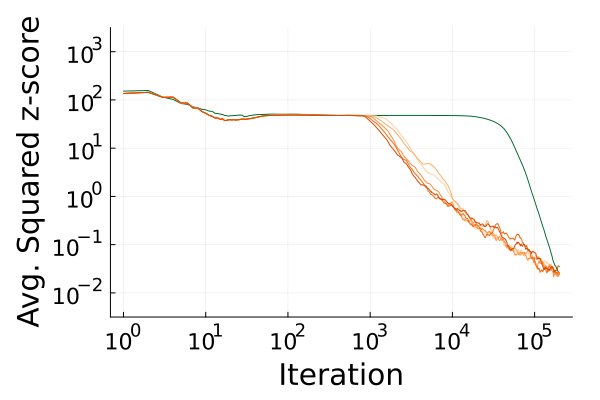}
    \caption{Logistic regression}
    \end{subfigure}
    \begin{subfigure}{0.33\textwidth}
    \includegraphics[width=\columnwidth]{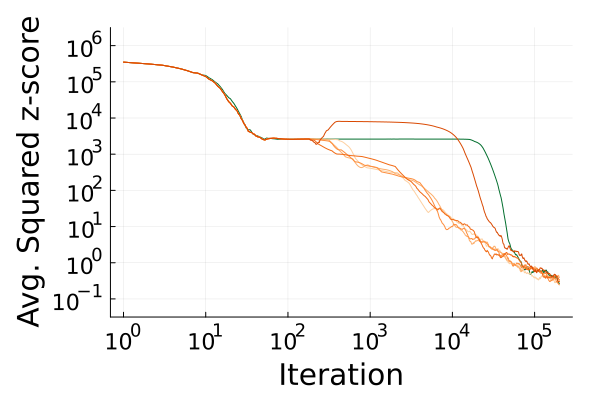}
    \caption{Poisson regression}
    \end{subfigure}
    \begin{subfigure}{0.33\textwidth}
    \includegraphics[width=\columnwidth]{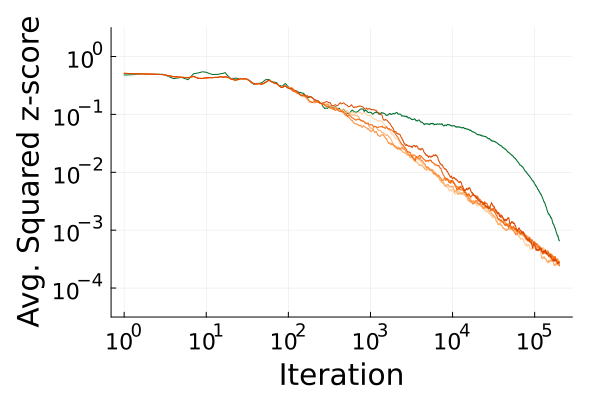}
    \caption{Bradley-Terry}
    \end{subfigure}
    \caption{Traces of average squared coordinate-wise z-scores between the true and approximated posterior 
    across all experiments, obtained from Hot DoG and optimally-tuned ADAM. 
    All figures share the legend in \cref{fig:tracecombined-linear}.
    The coreset size $M=1000$ and each line represents a different initial learning rate parameter. 
    The lines indicate the median from $10$ runs.}
    \label{fig:tracecombined}
\end{figure*}

\begin{figure*}
    \begin{subfigure}{0.24\textwidth}
        \includegraphics[width=\columnwidth]{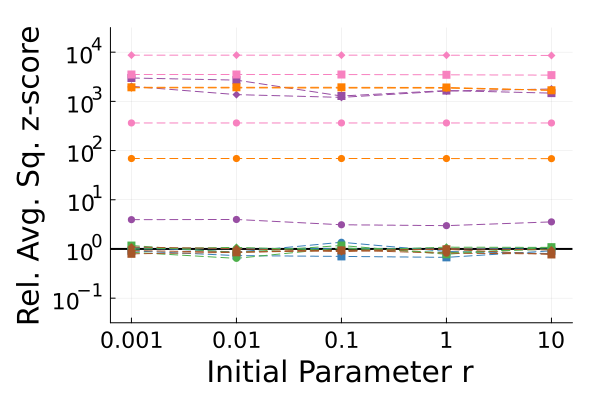}
        \includegraphics[width=\columnwidth]{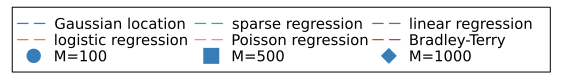}
        \caption{DoG}
    \end{subfigure}
    \begin{subfigure}{0.24\textwidth}
        \includegraphics[width=\columnwidth]{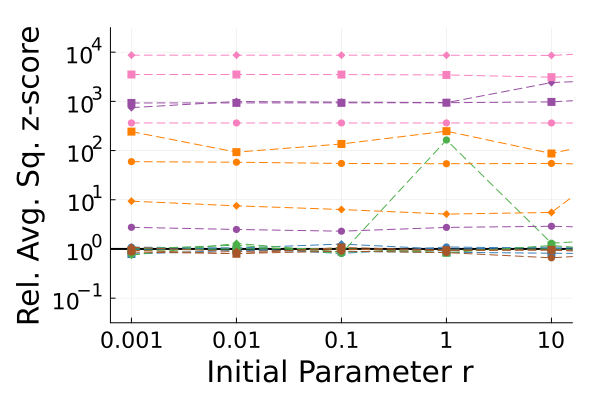}
        \includegraphics[width=\columnwidth]{plots/legend.png}
        \caption{DoWG}
    \end{subfigure}
    \begin{subfigure}{0.24\textwidth}
        \includegraphics[width=\columnwidth]{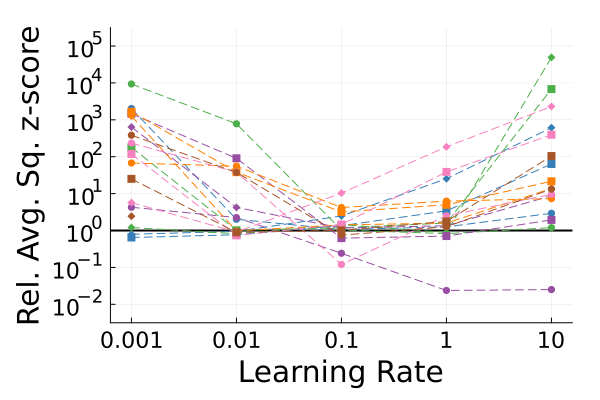}
        \includegraphics[width=\columnwidth]{plots/legend.png}
        \caption{ADAM}
    \end{subfigure}
    \begin{subfigure}{0.24\textwidth}
        \includegraphics[width=\columnwidth]{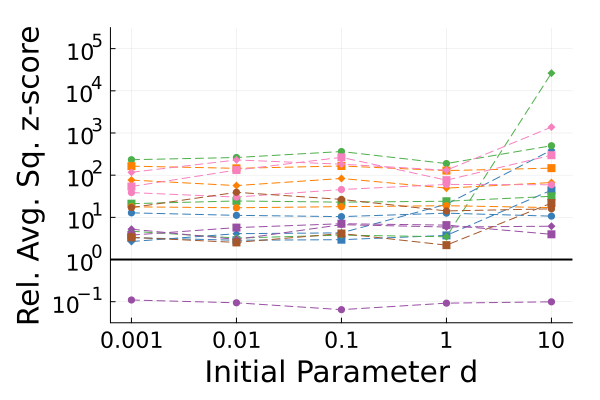}
        \includegraphics[width=\columnwidth]{plots/legend.png}
        \caption{prodigy ADAM}
    \end{subfigure}
    \begin{subfigure}{0.24\textwidth}
        \includegraphics[width=\columnwidth]{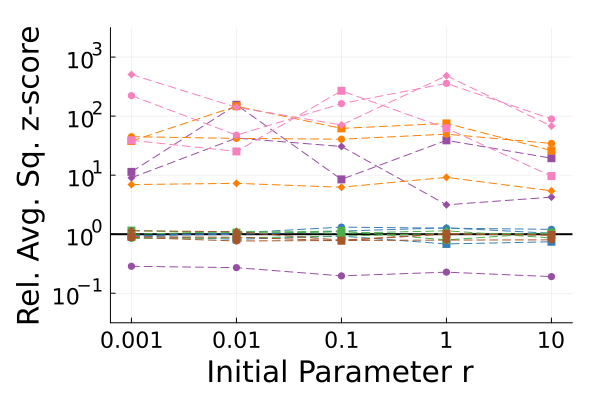}
        \includegraphics[width=\columnwidth]{plots/legend.png}
        \caption{DoG with hot-start}
    \end{subfigure}
    \begin{subfigure}{0.24\textwidth}
        \includegraphics[width=\columnwidth]{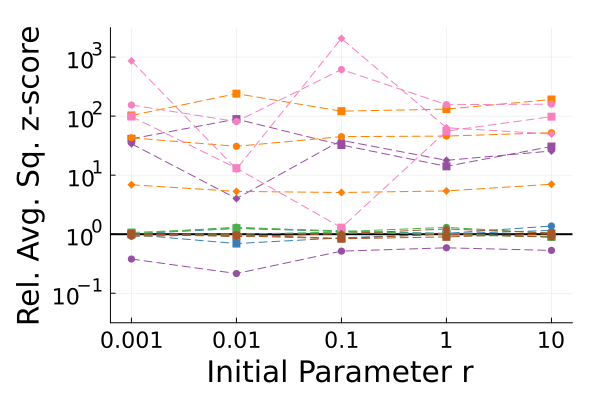}
        \includegraphics[width=\columnwidth]{plots/legend.png}
        \caption{DoWG with hot-start}
    \end{subfigure}
    \begin{subfigure}{0.24\textwidth}
        \includegraphics[width=\columnwidth]{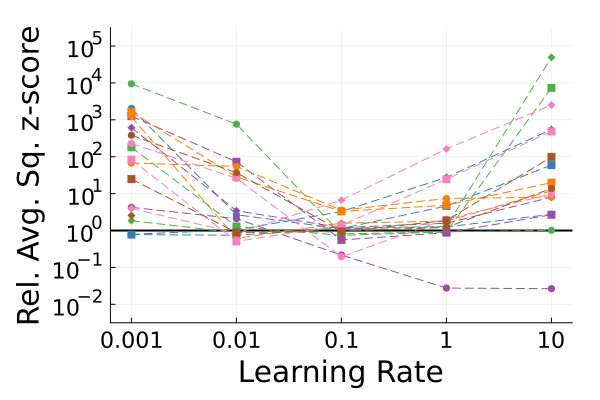}
        \includegraphics[width=\columnwidth]{plots/legend.png}
        \caption{ADAM with hot-start}
    \end{subfigure}
    \begin{subfigure}{0.24\textwidth}
        \includegraphics[width=\columnwidth]{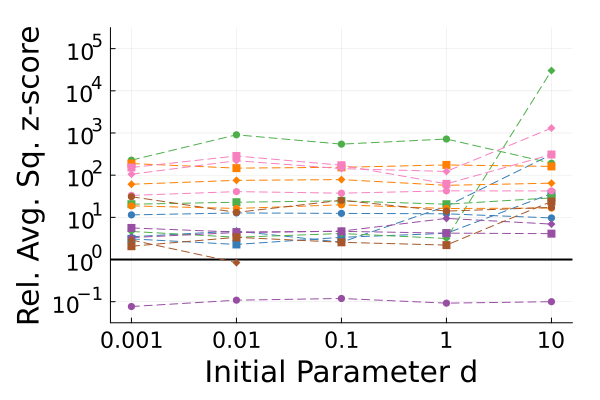}
        \includegraphics[width=\columnwidth]{plots/legend.png}
        \caption{prodigy ADAM with hot-start}
    \end{subfigure}
    \caption{Relative Coreset MCMC posterior approximation error comparing different optimization algorithms 
    (labeled in the subfigure captions) and the proposed Hot DoG method (with fixed $r=0.001$ and $c=0.5$).
    The metric plotted is
    the ratio of average squared z-scores (defined in \cref{eq:avg_sq_z}) under the algorithm labeled in each 
    subfigure caption to those under Hot DoG.
    Values above the horizontal black line ($10^0$) indicate that 
    the proposed Hot DoG method outperformed the method it compared to.
    Median values after $200,000$ optimization iterations across $10$ trials are used for the relative comparison 
    for a variety of datasets, models, and coreset sizes.}
    \label{fig:comparison_relative}
\end{figure*}

In this section, we demonstrate the effectiveness of Hot DoG and compare our method against other learning-rate-free 
stochastic gradient methods: optimally-tuned ADAM from a log scale grid search, as well as prodigy ADAM 
 \citep{mishchenko2023prodigy}, DoG \citep{ivgi2023dog}, and DoWG \citep{khaled2023dowg} over different initial parameters. 
We compare the quality of posterior approximations over different coreset sizes $M$ and weight optimization procedures.
Following \citet{chen2024coreset},
we set the number of Markov chains to $K=2$
and subsample size to $S=M$ in \cref{eq:gradest}.
We set
$\kappa_w$ to the hit-and-run slice sampler with doubling
\citep{belisle1993hit,neal2003slice} for all real data experiments.
For the Gaussian location model, we use a kernel that directly samples from $\pi_w$ [\citealp[Sec.~3.4]{chen2024coreset}];
for the sparse regression example, we use Gibbs sampling \citep{george1993variable}.

We compare these algorithms using six different Bayesian 
models, the details of which are in \cref{sec:appendix_4}.
We use Stan \citep{carpenter2017stan} to obtain full data inference results for real data experiments, 
and Gibbs sampling \citep{george1993variable} for the sparse regression model with discrete variables. 
For all experiments, we measure the posterior approximation quality using the 
average squared z-score, which we define as 
\[
    \frac{1}{D}\sum_{i=1}^D ( \frac{\mu_i - \hat{\mu}_i}{\sigma_i})^2. \label{eq:avg_sq_z}
\] 
In the above definition, $D$ denotes the dimension of $\Theta$; 
$\mu_i$ and $\sigma_i$ are, respectively, the coordinate-wise mean and standard deviation estimated using the full data posterior, 
and $\hat{\mu}_i$ is the coordinate-wise mean estimated using draws from Coreset MCMC.
This estimate is computed in a streaming fashion using the second half of all draws 
at the time; note this includes draws from $\pi_{w_0}$ before the hot-start test passes.

Each algorithm was run on 8 single-threaded cores of a 2.1GHz Intel Xeon Gold 6130 processor with 32GB memory. 
Code for these experiments is available at \url{https://github.com/NaitongChen/automated-coreset-mcmc-experiments}.
More experimental details and additional plots are in \cref{sec:appendix_4,sec:appendix_5}.

\textbf{Effect of hot-start test.}
\cref{fig:burnincomparison} compares Hot DoG with and without the hot-start test for $M=1000$ across all experiments;
the same plots for other coreset sizes can be found in \cref{sec:appendix_5}. 
Without the hot-start test, the traces often hit a long plateau, before the effect of 
exponentially-weighted averaging is able to decay early large gradient norms. 
On the other hand, with burn-in, we begin by simulating from Markov chains 
targeting $\pi_{w_0}$, and start optimizing the coreset weights only after the hot-start test has passed. 
In terms of the number of log potential evaluations, Hot DoG with 
burn-in leaves the plateau sooner than without burn-in. % phase.

\cref{fig:burnintest} examines the behaviour of the hot-start test in more detail, showing the traces 
of the gradient estimate norms $\|\hat{g}_t\|$ and test statistics \texttt{median}$(u_1,\dots,u_K)$ across optimization 
iterations when using Hot DoG. 
Here we only show plots for $M=1000$; the same plots for other 
coreset sizes can be found in \cref{sec:appendix_5}. 
In some experiments, the Markov chains are initialized reasonably well where 
the gradient norms are already stabilized, and the test passes almost immediately.
In others, the Markov chains are initialized poorly 
and the gradient norms are large, but nevertheless, the hot-start test passes
shortly after they stabilize. Across all 
experiments,
a test statistic threshold of 0.5 worked well.

\textbf{Robustness to fixed parameter $r$.}
\Cref{fig:tracecombined} provides an examination of the
robustness of the proposed method to the fixed initial learning rate parameter $r$. 
Across all experiments, different values of $r$ spanning multiple orders of magnitude 
result in similar posterior approximations across optimization iterations. Note that  $M$ is $1000$ for all plots in 
\cref{fig:tracecombined}. The same trends can be observed over different coreset sizes (see \cref{sec:appendix_5}).
In practice, we follow the recommendation of \citet{ivgi2023dog} and set $r=0.001$. 

\textbf{Comparison with other related methods.}
\Cref{fig:comparison_relative} shows a comparison between our method and DoG, DoWG, ADAM, as well as prodigy ADAM.
We fix $r=0.001$ and $c=0.5$ for Hot DoG. 
Since the hot-start test itself can be applied to all methods, Hot DoG is compared against 
others both with and without burn-in. The posterior approximation quality of 
Hot DoG is orders of magnitude better than all other methods in many settings tested, and remain competitive otherwise. 
In particular, Hot DoG is capable of matching the performance of optimally-tuned ADAM without tuning.
\section{Conclusion}
\label{sec:conclusion}
This paper introduced Hot DoG, a learning-rate-free stochastic gradient method
designed for learning coreset weights using Coreset MCMC. Our
method extends DoG, but includes adjustments tailored to the Markovian setting
of Coreset MCMC. In particular, Hot DoG includes a hot-start test detecting
when to start training coreset weights as well as acceleration
techniques. Our method is shown to produce coreset weights that converge to the 
optimum. The quality of coresets constructed by Hot DoG 
and their corresponding posterior approximation are robust to input parameters. 
Empirically, Hot DoG under our recommended setting ($r=0.001$ and $c=0.5$) 
produces better posterior approximations than other learning-rate-free stochastic 
gradient methods, and is competitive to optimally-tuned ADAM.

% \begin{contributions} % will be removed in pdf for initial submission 
% 					  % (without ‘accepted’ option in \documentclass)
%                       % so you can already fill it to test with the
%                       % ‘accepted’ class option
%     Briefly list author contributions. 
%     This is a nice way of making clear who did what and to give proper credit.
%     This section is optional.

%     H.~Q.~Bovik conceived the idea and wrote the paper.
%     Coauthor One created the code.
%     Coauthor Two created the figures.
% \end{contributions}

\begin{acknowledgements} % will be removed in pdf for initial submission,
						 % (without ‘accepted’ option in \documentclass)
                         % so you can already fill it to test with the
                         % ‘accepted’ class option
    T.C. and N.C. were supported by an NSERC Discovery Grant RGPIN-2019-03962, 
    and J.H.H. was partially supported by a National Science Foundation CAREER award IIS-2340586.
    We acknowledge the use of the ARC Sockeye computing platform from the University of British Columbia.
\end{acknowledgements}

% References
% \bibliographystyle{unsrtnat}
\bibliography{main.bib}

\newpage

\onecolumn

\title{Tuning-Free Coreset Markov Chain Monte Carlo via Hot DoG\\(Supplementary Material)}
\maketitle

\appendix
\section{Details of Experiments}
\label{sec:appendix_4}
\subsection{Model Specification}
In this subsection, we describe the six examples (two synthetic and four real data) that we used for our experiments. 
Processed versions of all datasets used for the experiments are available at 
\url{https://github.com/NaitongChen/automated-coreset-mcmc-experiments}.
For each of the regression models, we are given a set of points $(x_n,y_n)^N_{n=1}$, each consisting of features 
$x_n\in\reals^p$ and response $y_n$.

\textbf{Bayesian sparse linear regression:} This is based on Example 4.1 from \cite{george1993variable}. We use the model
\[
  \sigma^2 &\dist \InvGam\left(\nu/2, \nu\lambda/2\right),\\
  \forall i\in[p], \quad \gamma_i &\distiid \Bern(q), \\
  \beta_i \mid \gamma_i &\distind \Norm\left(0, \left(\mathds{1}(\gamma_i = 0)\tau + \mathds{1}(\gamma_i = 1)c\tau\right)^2 \right), \\ 
  \forall n\in[N], \quad y_n \mid x_n, \beta, \sigma^2 &\distind \Norm\left( x_n^\top\beta, \sigma^2 \right),
\]
where we set $\nu=0.1, \lambda=1, q=0.1, \tau=0.1$, and $c=10$. Here we model the variance $\sigma^2$, the vector of regression 
coefficients $\beta = \begin{bmatrix}\beta_1 & \dots & \beta_p \end{bmatrix}^\top\in\reals^p$ and  
the vector of binary variables $\gamma = \begin{bmatrix}\gamma_1 & \dots & \gamma_p \end{bmatrix}^\top\in\left\{0,1\right\}^p$ 
indicating the inclusion of the $p^\text{th}$ feature in the model.
We set $N=50{,}000$, $p=10$, $\beta^\star = \begin{bmatrix}0 & 0 & 0 & 0 & 0 & 5 & 5 & 5 & 5 & 5\end{bmatrix}^\top$, 
and generate a synthetic dataset by
\[
  \forall n \in [N], \quad x_n &\distiid \Norm\left( 0, I \right),\\
  \eps_n &\distiid \Norm\left( 0, 25^2 \right),\\
  y_n &= x_n^\top \beta^\star + \eps_n.
\]

\textbf{Bayesian linear regression:} We use the model
\[
  \begin{bmatrix} \beta & \log\sigma^2 \end{bmatrix}^\top &\dist \Norm(0,I),\\ 
  \forall n\in[N],
  y_n \mid x_n, \beta, \sigma^2 &\distind \Norm\left(\begin{bmatrix} 1 & x_n^\top \end{bmatrix}\beta, \sigma^2\right),
\]
where $\beta\in\reals^{p+1}$ is a vector of regression coefficients and $\sigma^2\in\reals_+$ is the noise variance. 
Note that the prior here is not conjugate for the likelihood.
The dataset consists of flight delay information from $N= 98{,}673$ observations and was constructed using flight 
delay data from \url{https://www.transtats.bts.gov/Homepage.asp} 
and historical weather information from \url{https://www.wunderground.com/}. We study 
the difference, in minutes, between the scheduled and actual departure times against $p=10$ features 
including flight-specific and meteorological information.

\textbf{Bayesian logistic regression:} We use the model
\[
    \forall i\in[p+1], \quad \beta_i &\distiid \Cauchy(0,1), \\
	\forall n\in[N], \quad y_n &\distind \Bern 
    \left(\left(1+\exp\left(-\begin{bmatrix} 1 & x_n^\top\end{bmatrix} \beta\right)\right)^{-1}\right),  
\]
where $\beta = \begin{bmatrix}\beta_1 & \dots & \beta_{p+1} \end{bmatrix}^\top\in\reals^{p+1}$ is a 
vector of regression coefficients. Here we use the same dataset as in linear 
regression, but instead model the relationship between whether a flight is cancelled using the same set of features. 
Note that of all flights included, only $0.058\%$ were cancelled.

\textbf{Bayesian Poisson regression:} We use the model
\[
  \beta &\dist \Norm(0, I),\\
  \forall n\in[N],
  y_n \mid x_n, \beta &\distind 
  \Poiss\left( \log\left( 1 + e^{ \begin{bmatrix} 1 & x_n^\top \end{bmatrix}\beta } \right) \right),
\]
where $\beta\in\reals^{p+1}$ is a vector of regression coefficients. 
The dataset consists of $N= 15{,}641$ observations, and 
we model the hourly count of rental bikes against $p=8$ features (e.g., temperature, humidity at the 
time, and whether or not the day is a workday). The original bike share dataset is available at 
\url{https://archive.ics.uci.edu/dataset/275/bike+sharing+dataset}.

The remaining two non-regression models are specified as follows.

\textbf{Gaussian location:} We use the model
\[
\theta &\dist \Norm(0, I),\\
\forall n \in [N], X_n &\distiid \Norm(\theta, I),
\]
where $\theta, X_n\in\reals^d$. Here we model the mean $\theta$. We set $N=10{,}000, d=20$ and generate a synthetic dataset by
\[
    \forall n\in[N], x_n\distiid\Norm(0,I).
\]

\textbf{Bradley-Terry model:} We use the model
\[
    \theta &\distiid \Norm(0, I),\\
    \forall n\in[N], y_{n} \mid h_n, v_n, \theta &\distind \Bern\left( \left( 1+\exp\left( (\theta_{v_n} - \theta_{h_n})/400 \right) \right)^{-1} \right),
\]
where $\theta \in \reals^d$. 
The dataset was constructed using games statistics from \url{https://www.nba.com/stats}
and consists of data of $N=26,651$ NBA games between the 2004 and 2022 seasons. 
$h_n$ and $v_n$ are the home team and visitor team IDs for the $n^\text{th}$ game in the dataset, and $y_n$ denotes 
the outcome of the game ($y_n=1$ if the home team won and $y_n=0$ if the visitor team won).
$\theta\in\reals^d$ represents the Elo ratings or relative skill levels [\citealp[Ch.~1]{elo1978rating}] for each of 
the $d=30$ teams. We model the Elo ratings using outcomes of pairwise comparisons between teams using game outcomes.

\subsection{Parameter Settings}
For full-data inference results of all examples except for the sparse linear regression model, we ran Stan 
\citep{carpenter2017stan} with 10 parallel chains,
each taking $100{,}000$ steps with the first $50{,}000$ discarded, for a combined $500{,}000$ draws. 
For full-data inference result of the sparse linear regression example, we use the Gibbs sampler developed by \cite{george1993variable}
to generate $200{,}000$ draws, with the first half discarded as burn-in.

To account for changes in $w$, for all real data experiments, we use the hit-and-run slice sampler 
with doubling \citep{belisle1993hit,neal2003slice}; 
for the Gaussian location model, we use a kernel that directly samples from $\pi_w$ [\citealp[Sec.~3.4]{chen2024coreset}].
for the sparse regression, we use the Gibbs sampler developed by \cite{george1993variable}.

We use Stan \citep{carpenter2017stan} to obtain full data inference results for real data experiments, 
and Gibbs sampling \citep{george1993variable} for the sparse regression model with discrete variables. 
The true posterior distribution for the Gaussian location model is available in closed form.

For ADAM, we test multiple learning rates over a log scale grid 
$\left( 10^k \right)$ for $k=-3,-2,\dots,1$.
For each experiment under each coreset size, the optimally-tuned ADAM is the one that obtained the lowest average 
squared z-score after $200,000$ iterations of weight optimization. 
For all learning-rate-free methods, we test different initial parameters (initial lower bound for prodigy ADAM 
and $r_0$ for Hot DoG, DoG, and DoWG) over a log scaled grid $\left( 10^k \right)$ for $k=-3, -2, \dots, 1$.

For the logistic regression example, to account for the class imbalance problem, we include all observations from the 
rare positive class if the coreset size is more than twice as big as the total number of observations with 
positive labels. Otherwise we sample our coreset to have $50\%$ positive labels and $50\%$ negative labels. 
Coreset points are uniformly subsampled for all other models.

\section{Additional Results}
\label{sec:appendix_5}
Figs. 4 to 6 in the main text
show the traces of average squared coordinate-wise z-scores, 
as well as the gradient estimate norms and hot-start test statistics for Hot DoG when $M=1000$. In this subsection, we 
show the same sets of plots for $M=100$ and $M=500$. 
Similarly to \cref{fig:burnincomparison}, 
\cref{fig:burnincomparison100,fig:burnincomparison500} compare Hot DoG with and without hot-start test. 
Similarly to \cref{fig:burnintest}, 
\cref{fig:burnintest100,fig:burnintest500} show the gradient estimate norms and hot-start test statistics during burn-in.
Similarly to \cref{fig:tracecombined}, 
\cref{fig:tracecombined100,fig:tracecombined500} compare Hot DoG (with hot-start test) and optimally-tuned ADAM.
We see that all plots show the same trends as the ones in 
\cref{sec:experiments}, 
where $M=1000$.
As a result, we arrive at similar observations as in
\cref{sec:experiments}.
% Section 5. 
In particular, Hot DoG with burn-in leaves the plateau sooner than without burn-in; the hot-start test passes and 
thus burn-in terminates shortly after gradient norms are stabilized; Hot DoG is robust to the fixed parameter $r$.

\begin{figure}[h]
    \begin{subfigure}{0.33\textwidth}
    \includegraphics[width=\columnwidth]{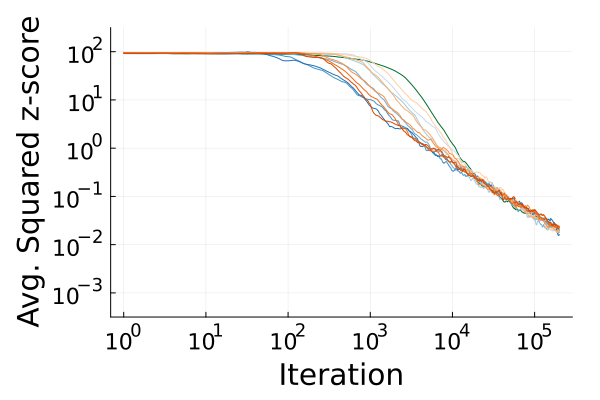}
    \caption{Gaussian location}
    \end{subfigure}
    \begin{subfigure}{0.33\textwidth}
    \includegraphics[width=\columnwidth]{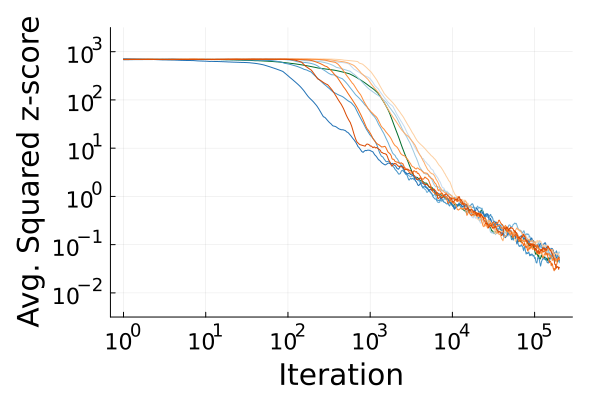}
    \caption{Sparse regression}
    \end{subfigure}
    \begin{subfigure}{0.33\textwidth}
    \includegraphics[width=\columnwidth]{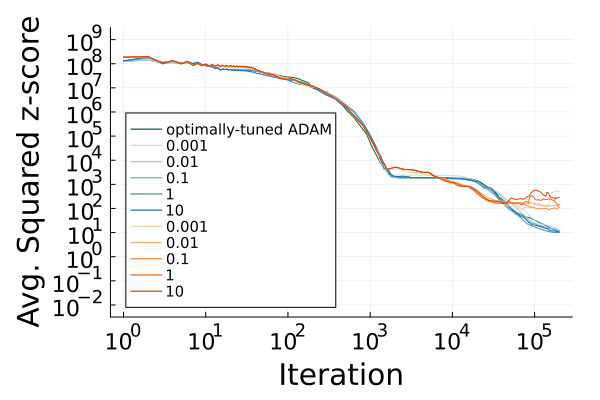}
    \caption{Linear regression}\label{fig:burnincomparison-linear100}
    \end{subfigure}
    \begin{subfigure}{0.33\textwidth}
    \includegraphics[width=\columnwidth]{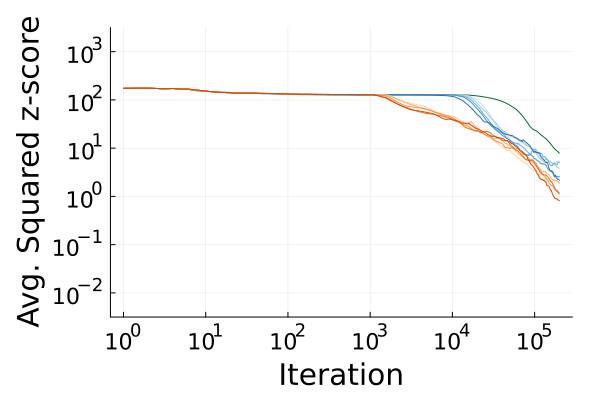}
    \caption{Logistic regression}
    \end{subfigure}
    \begin{subfigure}{0.33\textwidth}
    \includegraphics[width=\columnwidth]{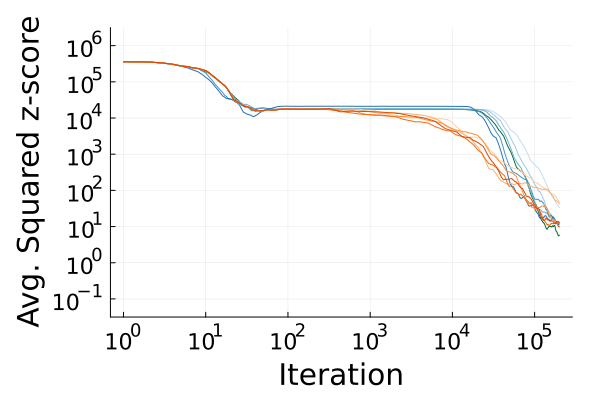}
    \caption{Poisson regression}
    \end{subfigure}
    \begin{subfigure}{0.33\textwidth}
    \includegraphics[width=\columnwidth]{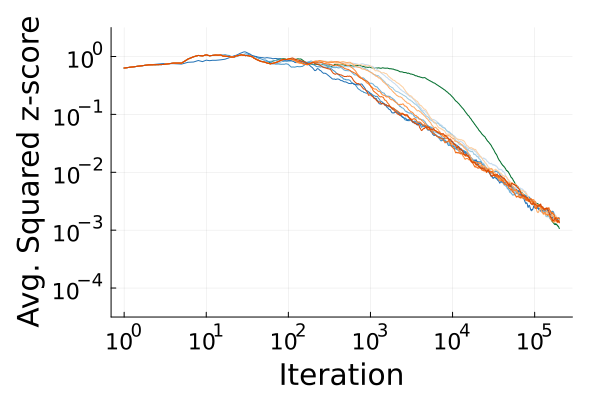}
    \caption{Bradley-Terry}
    \end{subfigure}
    \caption{Traces of average squared coordinate-wise z-scores between the true and approximated posterior 
    across all experiments, obtained using Hot DoG with and without hot-start test. 
    All figures share the legend in \cref{fig:burnincomparison-linear100}. The coreset size $M$ is $100$ and each line
    represents a different initial learning rate parameter. The lines indicate the median from $10$ runs.
    Orange lines indicate runs with hot-start test and blue lines without.}
    \label{fig:burnincomparison100}
\end{figure}

\begin{figure}[h]
    \begin{subfigure}{0.33\textwidth}
    \includegraphics[width=\columnwidth]{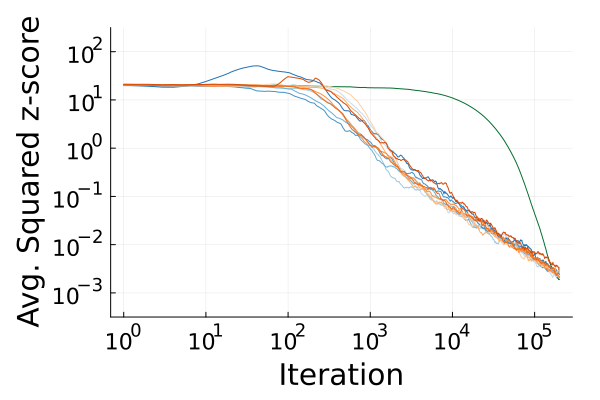}
    \caption{Gaussian location}
    \end{subfigure}
    \begin{subfigure}{0.33\textwidth}
    \includegraphics[width=\columnwidth]{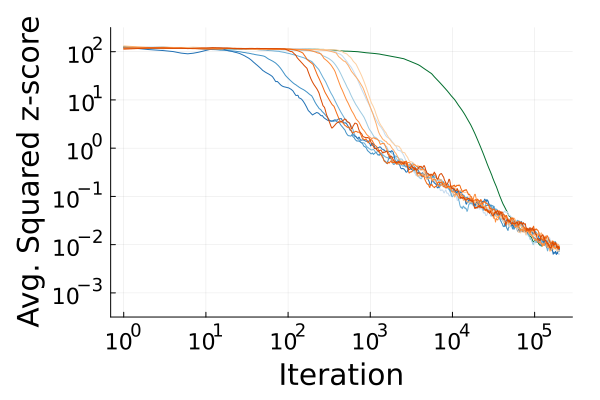}
    \caption{Sparse regression}
    \end{subfigure}
    \begin{subfigure}{0.33\textwidth}
    \includegraphics[width=\columnwidth]{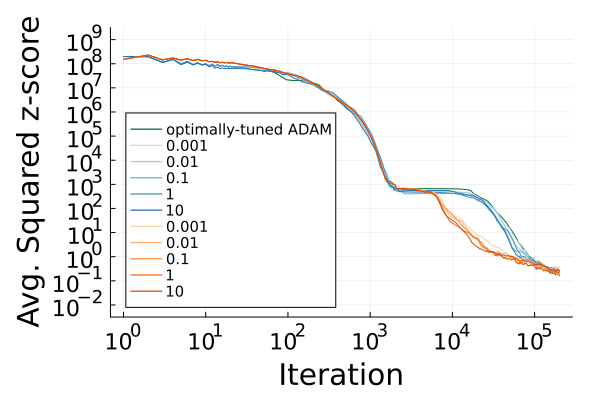}
    \caption{Linear regression}\label{fig:burnincomparison-linear500}
    \end{subfigure}
    \begin{subfigure}{0.33\textwidth}
    \includegraphics[width=\columnwidth]{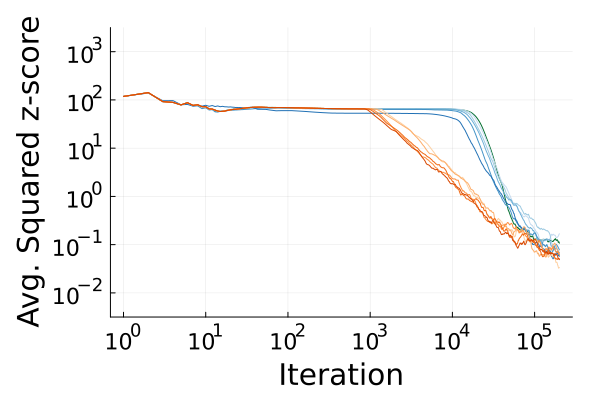}
    \caption{Logistic regression}
    \end{subfigure}
    \begin{subfigure}{0.33\textwidth}
    \includegraphics[width=\columnwidth]{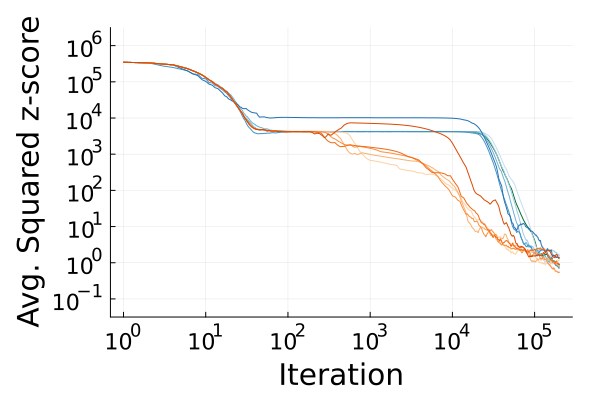}
    \caption{Poisson regression}
    \end{subfigure}
    \begin{subfigure}{0.33\textwidth}
    \includegraphics[width=\columnwidth]{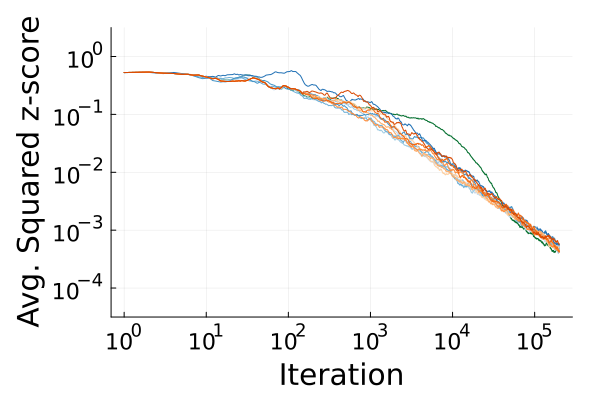}
    \caption{Bradley-Terry}
    \end{subfigure}
    \caption{Traces of average squared coordinate-wise z-scores between the true and approximated posterior 
    across all experiments, obtained using Hot DoG with and without hot-start test. 
    All figures share the legend in \cref{fig:burnincomparison-linear500}. The coreset size $M$ is $500$ and each line
    represents a different initial learning rate parameter. The lines indicate the median from $10$ runs.
    Orange lines indicate runs with hot-start test and blue lines without.}
    \label{fig:burnincomparison500}
\end{figure}

\begin{figure}[h]
    \begin{subfigure}{0.33\textwidth}
        \includegraphics[width=\columnwidth]{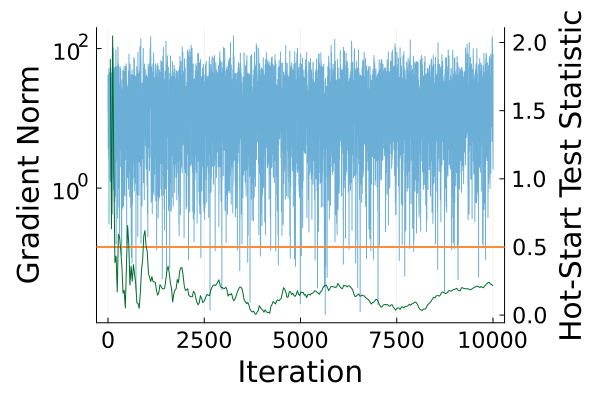}
        \caption{Gaussian location}
    \end{subfigure}
    \begin{subfigure}{0.33\textwidth}
        \includegraphics[width=\columnwidth]{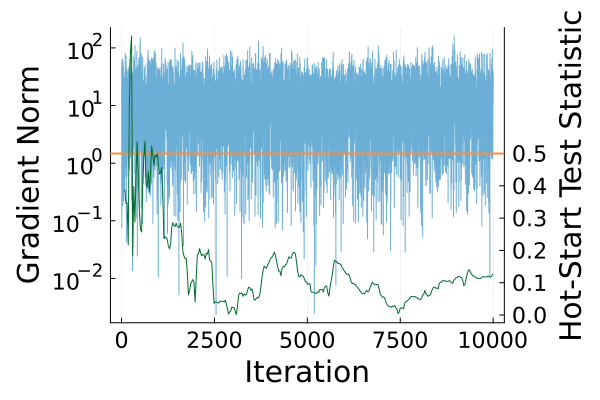}
        \caption{Sparse regression}
    \end{subfigure}
    \begin{subfigure}{0.33\textwidth}
        \includegraphics[width=\columnwidth]{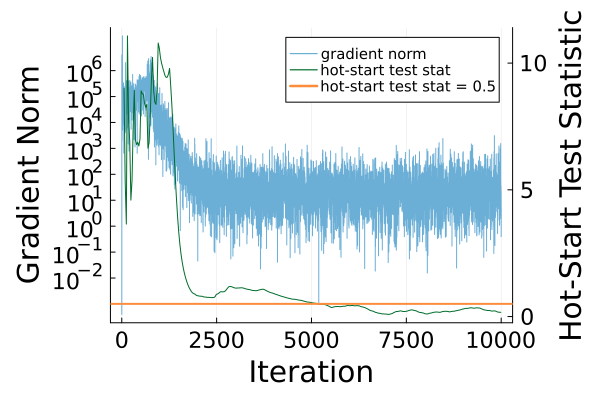}
        \caption{Linear regression}
    \end{subfigure}
    \begin{subfigure}{0.33\textwidth}
        \includegraphics[width=\columnwidth]{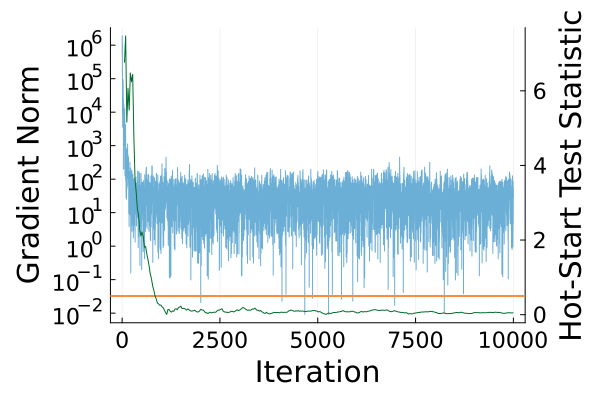}
        \caption{Logistic regression}
    \end{subfigure}
    \begin{subfigure}{0.33\textwidth}
        \includegraphics[width=\columnwidth]{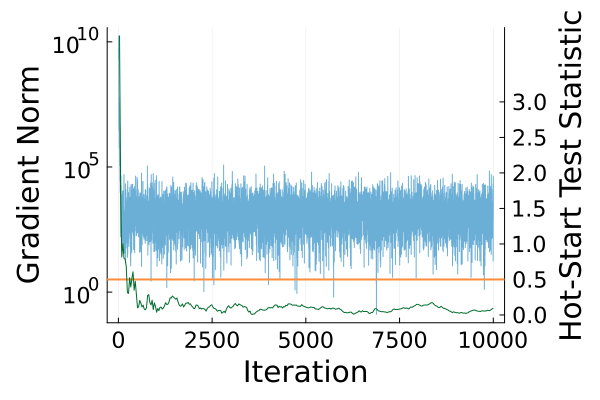}
        \caption{Poisson regression}
    \end{subfigure}
    \begin{subfigure}{0.33\textwidth}
        \includegraphics[width=\columnwidth]{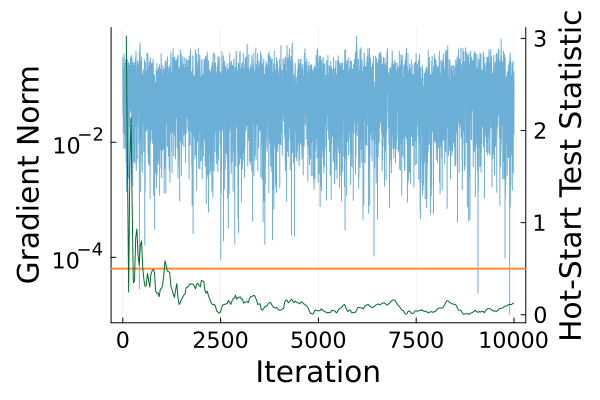}
        \caption{Bradley-Terry}
    \end{subfigure}
    \caption{Trace of gradient estimate norms (blue) and hot-start test statistics (green) before weight optimization
            across all experiments with $M=100$.
            The orange horizontal line is the test statistic threshold $c=0.5$.}
    \label{fig:burnintest100}
\end{figure}

\begin{figure}[h]
    \begin{subfigure}{0.33\textwidth}
        \includegraphics[width=\columnwidth]{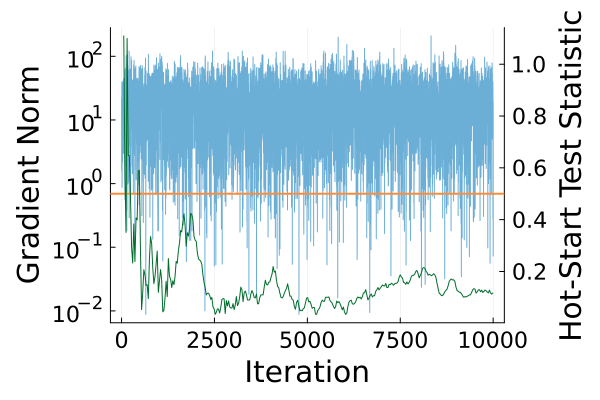}
        \caption{Gaussian location}
    \end{subfigure}
    \begin{subfigure}{0.33\textwidth}
        \includegraphics[width=\columnwidth]{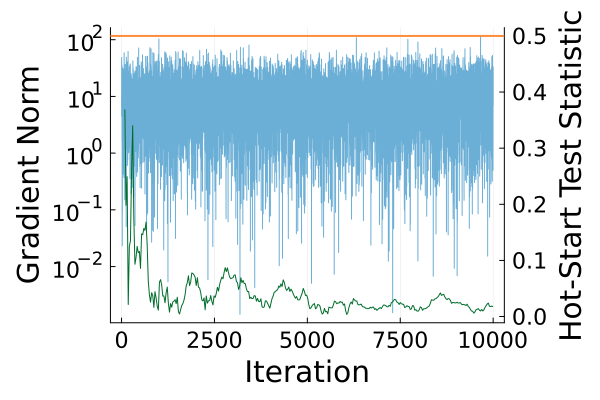}
        \caption{Sparse regression}
    \end{subfigure}
    \begin{subfigure}{0.33\textwidth}
        \includegraphics[width=\columnwidth]{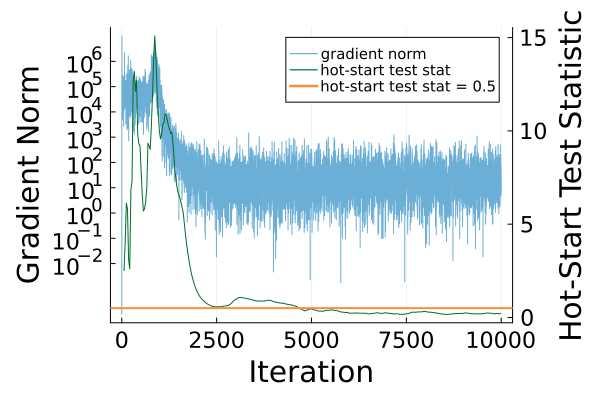}
        \caption{Linear regression}
    \end{subfigure}
    \begin{subfigure}{0.33\textwidth}
        \includegraphics[width=\columnwidth]{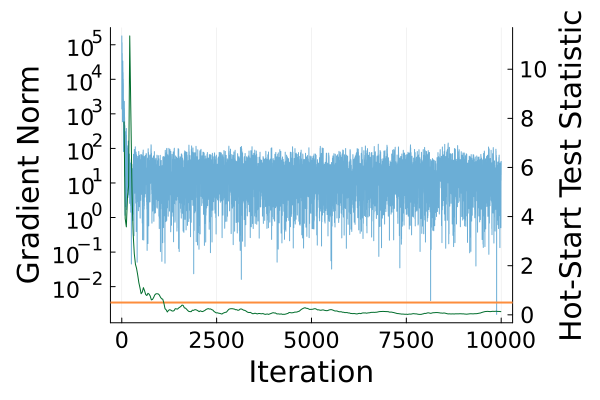}
        \caption{Logistic regression}
    \end{subfigure}
    \begin{subfigure}{0.33\textwidth}
        \includegraphics[width=\columnwidth]{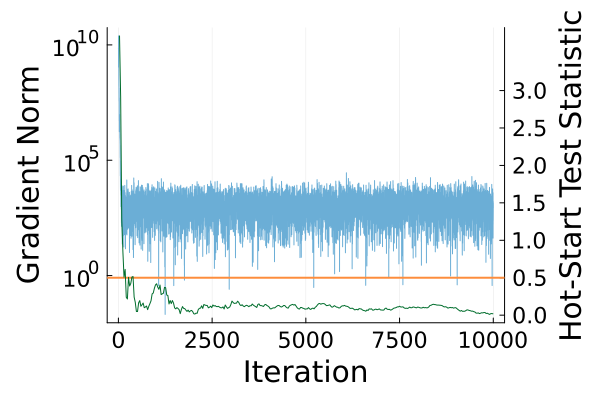}
        \caption{Poisson regression}
    \end{subfigure}
    \begin{subfigure}{0.33\textwidth}
        \includegraphics[width=\columnwidth]{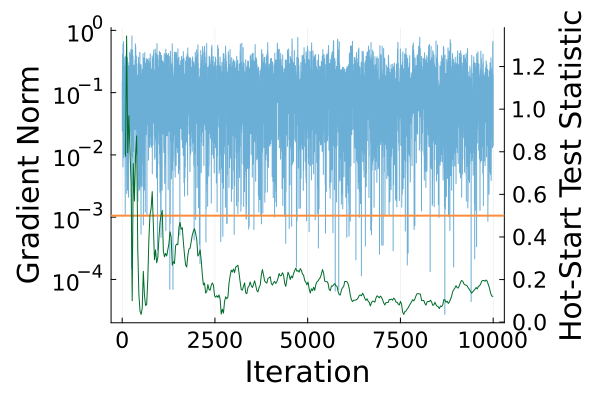}
        \caption{Bradley-Terry}
    \end{subfigure}
    \caption{Trace of gradient estimate norms (blue) and hot-start test statistics (green) before weight optimization
            across all experiments with $M=500$.
            The orange horizontal line is the test statistic threshold $c=0.5$.}
    \label{fig:burnintest500}
\end{figure}

\begin{figure}[h]
    \begin{subfigure}{0.33\textwidth}
    \includegraphics[width=\columnwidth]{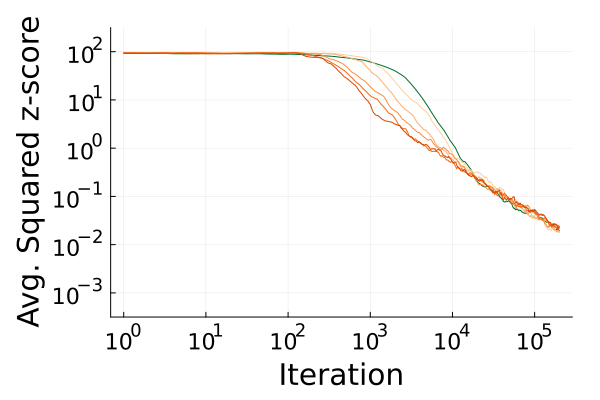}
    \caption{Gaussian location}
    \end{subfigure}
    \begin{subfigure}{0.33\textwidth}
    \includegraphics[width=\columnwidth]{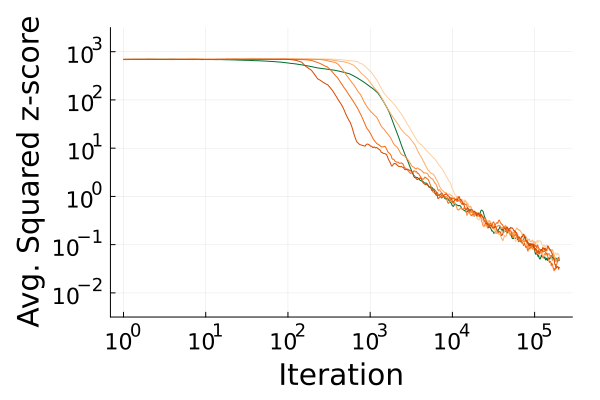}
    \caption{Sparse regression}
    \end{subfigure}
    \begin{subfigure}{0.33\textwidth}
    \includegraphics[width=\columnwidth]{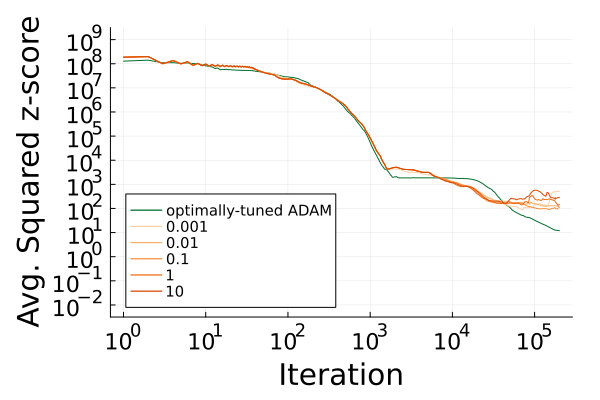}
    \caption{Linear regression}\label{fig:tracecombined-linear100}
    \end{subfigure}
    \begin{subfigure}{0.33\textwidth}
    \includegraphics[width=\columnwidth]{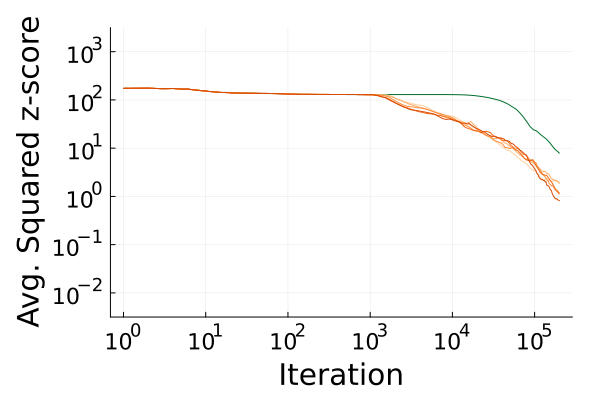}
    \caption{Logistic regression}
    \end{subfigure}
    \begin{subfigure}{0.33\textwidth}
    \includegraphics[width=\columnwidth]{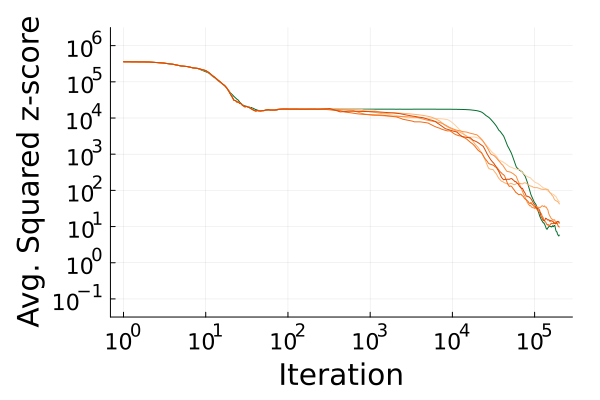}
    \caption{Poisson regression}
    \end{subfigure}
    \begin{subfigure}{0.33\textwidth}
    \includegraphics[width=\columnwidth]{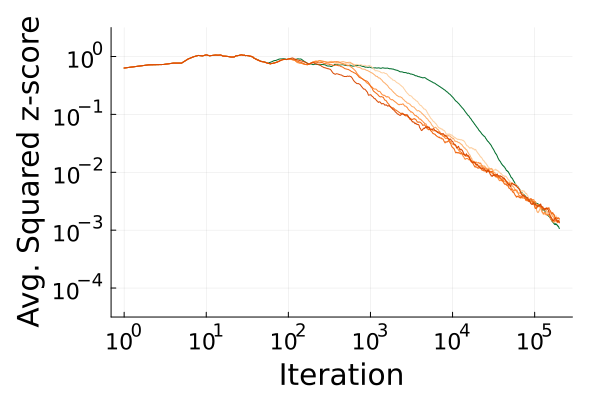}
    \caption{Bradley-Terry}
    \end{subfigure}
    \caption{Traces of average squared coordinate-wise z-scores between the true and approximated posterior 
    across all experiments, obtained from Hot DoG and optimally-tuned ADAM. 
    All figures share the legend in \cref{fig:tracecombined-linear100}.
    The coreset size $M=100$ and each line represents a different initial learning rate parameter. 
    The lines indicate the median from $10$ runs.}
    \label{fig:tracecombined100}
\end{figure}

\begin{figure}[h]
    \begin{subfigure}{0.33\textwidth}
    \includegraphics[width=\columnwidth]{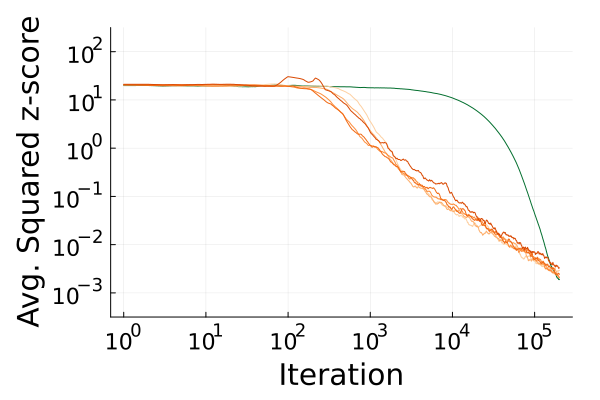}
    \caption{Gaussian location}
    \end{subfigure}
    \begin{subfigure}{0.33\textwidth}
    \includegraphics[width=\columnwidth]{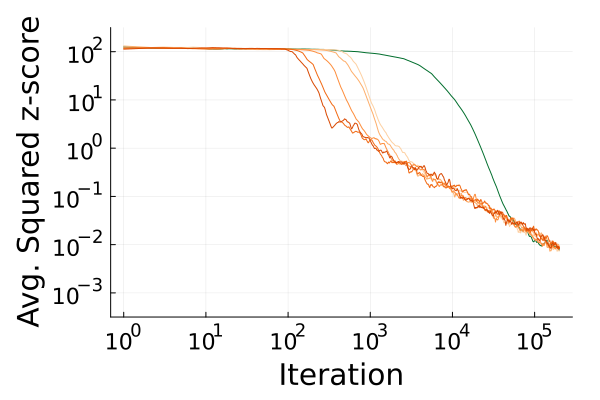}
    \caption{Sparse regression}
    \end{subfigure}
    \begin{subfigure}{0.33\textwidth}
    \includegraphics[width=\columnwidth]{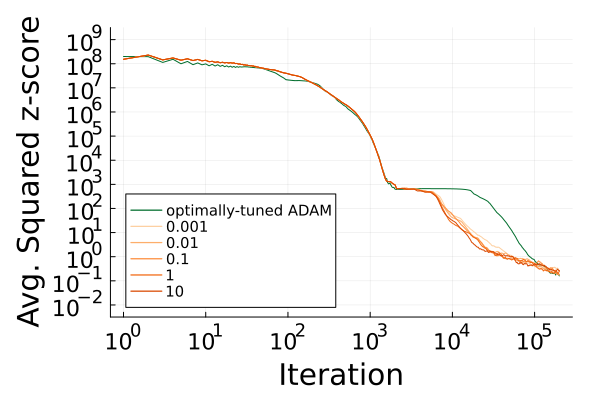}
    \caption{Linear regression}\label{fig:tracecombined-linear500}
    \end{subfigure}
    \begin{subfigure}{0.33\textwidth}
    \includegraphics[width=\columnwidth]{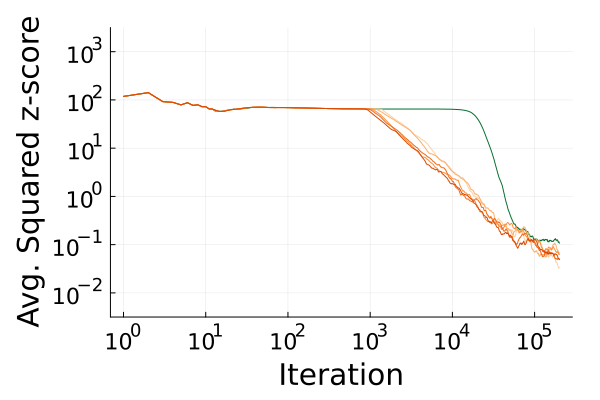}
    \caption{Logistic regression}
    \end{subfigure}
    \begin{subfigure}{0.33\textwidth}
    \includegraphics[width=\columnwidth]{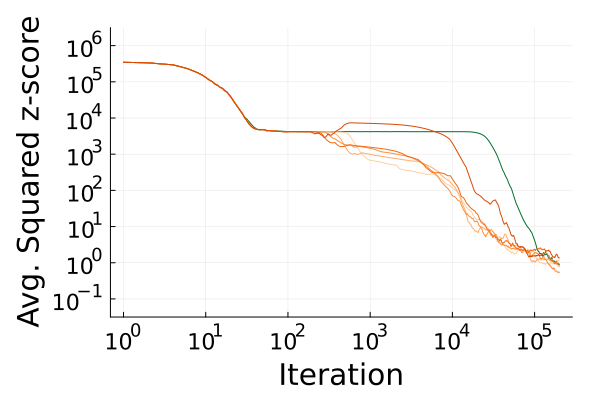}
    \caption{Poisson regression}
    \end{subfigure}
    \begin{subfigure}{0.33\textwidth}
    \includegraphics[width=\columnwidth]{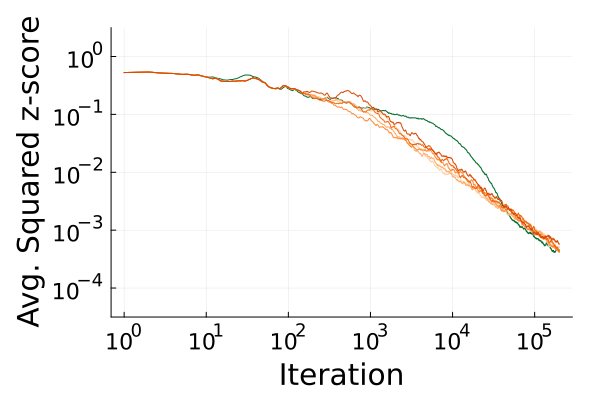}
    \caption{Bradley-Terry}
    \end{subfigure}
    \caption{Traces of average squared coordinate-wise z-scores between the true and approximated posterior 
    across all experiments, obtained from Hot DoG and optimally-tuned ADAM. 
    All figures share the legend in \cref{fig:tracecombined-linear500}.
    The coreset size $M=500$ and each line represents a different initial learning rate parameter. 
    The lines indicate the median from $10$ runs.}
    \label{fig:tracecombined500}
\end{figure}
\clearpage
\section{Hot DoG Convergence}\label{sec:convergence_proof}
\subsection{Update Rule}\label{sec:update_rule}
The Hot DoG update for coordinate $i \in [M]$ at iteration $n\in\nats_+$ can be written as
\[
    m_{t,i} &= \beta_1 m_{t-1, i} + g_i(w_{t-1}, \theta_{t-1}, \mathcal{S}_{t-1}) \label{eq:m_update}\\
    v_{t,i} &= \beta_2 v_{t-1, i} + \left(g_i(w_{t-1}, \theta_{t-1}, \mathcal{S}_{t-1})\right)^2 \label{eq:v_update}\\
    w_{t,i} &= w_{t-1,i} - \alpha_{t,i} \frac{m_{t,i}}{\sqrt{t \cdot (\epsilon + v_{t,i})}},
\]
where
\[
    \alpha_{t,i} &= \frac{1-\beta_1}{1-\beta_1^t} \frac{\sqrt{1-\beta_2^t}}{\sqrt{1-\beta_2}} \tilde{r}_{t,i}\\
    \tilde{r}_{t,i} &= (1-\beta_1^t)^{-1}\left((1-\beta_1)\left(\sum_{k=0}^{t-1}\beta_1^k\bar{r}_{t-k,i}\right) + \beta_1^t\bar{r}_{0,i}\right)\\
    \bar{r}_{t,i} &= \left(\max_{k\leq t}\{ |w_{t,i} - w_{0,i}| \}\right) \vee r_\delta.
\]
We initialize the algorithm such that $m_0=0$ and $v_0=0$.
Define $s_t \in \reals^N$ such that $\forall j \in \mathcal{S}_t, s_{tj} = \frac{N}{S}$ and $0$ otherwise.
The $M$-dimensional subsampled gradient estimate as defined in \cref{eq:gradest} then takes the form
\[
    g(w_{t-1}, \theta_{t-1}, \mathcal{S}_{t-1}) = G_{t-1}(w_{t-1} - w^\star) + H_{t-1}(1-s_{t-1}), \label{eq:gradient}
\]
where  
\[
    G_{t\!-\!1} \!=\! \frac{1}{K\!-\!1} \!\sum_{k=1}^K\! \begin{bmatrix} \bar{\ell}_1 (\theta_{(t\!-\!1)k}) \\ \vdots \\ \ell_M (\theta_{(t-1)k}) \end{bmatrix}
                                        \begin{bmatrix} \bar{\ell}_1 (\theta_{(t\!-\!1)k}) \\ \vdots \\ \ell_M (\theta_{(t-1)k}) \end{bmatrix}^\top \!\in\! \reals^{M\!\times\! M},
    H_{t\!-\!1} \!=\! \frac{1}{K\!-\!1}\!\sum_{k=1}^K\!
                \bbmat \bar\ell_1(\theta_{(t\!-\!1)k})\\ \vdots \\ \bar\ell_M(\theta_{(t-1)k})\ebmat
                \bbmat \bar\ell_1(\theta_{(t\!-\!1)k})\\ \vdots \\ \bar\ell_N(\theta_{(t-1)k})\ebmat^\top
                \!\in\! \reals^{M\!\times\! N}.
\]
Note that in \cref{eq:gradient}, both matrix-vector products on the right hand side give us vectors of dimension $M$, 
which aligns with the desired dimension of the gradient estimate.
We also define here two quantities that improves the readability of proofs presented in following subsections:
\[
    R_{t-1} = \begin{bmatrix}\frac{\alpha_{t,1}}{\sqrt{t\cdot (\epsilon + v_{t,1})}} & \cdots & \frac{\alpha_{t,M}}{\sqrt{t\cdot (\epsilon + v_{t,M})}}\end{bmatrix}^\top, \quad\quad
    \Delta_{t-j} = w_{t-j-1} - w_{t-j} = \alpha_{t-j} \odot \frac{m_{t-j}}{\sqrt{(t-j)\cdot \left(\epsilon + v_{t-j}\right)}}. \label{eq:simple_notation}
\]
\subsection{Assumptions}\label{sec:assumptions}
\begin{assumption}[Coreset weight constraint]\label{assump:constraint}
    $\mathcal{W} = \{w \in \reals^M: w_t \geq 0, \sum_{m=1}^M w_{tm} \leq B\}$.
\end{assumption}
\begin{assumption}[Exact coreset]\label{assump:exact}
    There exists a $w^\star \in \reals^M, c^\star \in \reals$ such that $w^\star \in \mathcal{W}$ and 
    \[
        \sum_{n=1}^N \ell_n(\cdot) = \sum_{m=1}^M w_m^\star \ell_m(\cdot) + c^\star \quad \pi_0 - a.e.v.
    \]
\end{assumption}
\begin{assumption}[Bounded gradient]\label{assump:grad_bound}
    There exists $U>0$ such that
    \[
        \forall w_t \in \mathcal{W}, \theta_t \in \Theta^K, \mathcal{S}_{t} \subseteq [N] \,\,\,\, 
        \|g(w_t,\theta_t, \mathcal{S}_t)\|_\infty \leq U.
    \]
\end{assumption}
\begin{assumption}[Markov gradient mixing]\label{assump:mixing}
    There exists $\lambda>0$ such that 
    \[
        \forall w_t \in \mathcal{W}, \theta_{t-1} \in \Theta^K \,\,\,\, 
        \E\left[ G_t \middle | w_t, \theta_{t-1} \right] \succeq \lambda I.
    \]
\end{assumption}
\begin{assumption}[Markov gradient noise boundedness]\label{assump:noise}
    There exists $0<\bar{\lambda}<\infty$ such that 
    \[
        \forall w_t, w_{t-j} \in \mathcal{W}, \theta_{t-1}, \theta_{t-j-i} \in \Theta^K \,\,\,\, 
        \E\left[ G_{t-j}^\top G_{t} \middle | w_t, \theta_{t-1}, w_{t-j}, \theta_{t-j-1} \right] \preceq \bar{\lambda} I.
    \]
\end{assumption}

\subsection{Convergence Proof}\label{sec:proof}
\begin{proof}[Proof of \cref{thm:convergence}]
We begin by applying the projected gradient update to get
\[
    \|w_t - w^\star\|^2 
    &= \left\| \proj_\mathcal{W} \left( w_{t-1} - \alpha_t \odot \frac{m_t}{\sqrt{t \cdot (\epsilon + v_t)}} \right) - w^\star \right\|^2\\
    &= \left\| \proj_\mathcal{W} \left( w_{t-1} - \alpha_t \odot \frac{m_t}{\sqrt{t \cdot (\epsilon + v_t)}} \right) - \proj_\mathcal{W} w^\star \right\|^2\\
    &\leq \left\| w_{t-1} - \alpha_t \odot \frac{m_t}{\sqrt{t \cdot (\epsilon + v_t)}} - w^\star \right\|^2. \label{eq:one_step_update}
\]
Here $\odot$ denotes element-wise multiplication, and the fraction $\frac{m_t}{\sqrt{t \cdot (\epsilon + v_t)}}$ is also applied element-wise. The second equality follows because
$w^\star \in \mathcal{W}$ by assumption. The inequality follows because $\mathcal{W}$ defined in \cref{assump:constraint} 
is convex and closed, and hence $\proj_{\mathcal{W}}$ is a contraction.
We unroll $m_t$ by \cref{eq:m_update} and use $R_{t-1}$ as defined in \cref{eq:simple_notation} to get
\[
    \alpha_t \odot \frac{m_t}{\sqrt{t \cdot (\epsilon + v_t)}} = 
    \diag(R_{t-1}) \sum_{k=0}^{t-1}\beta_1^k g(w_{t-k-1}, \theta_{t-k-1}, \mathcal{S}_{t-k-1}).\label{eq:element_update}
\]
By substituting \cref{eq:gradient,eq:element_update} into \cref{eq:one_step_update} and taking expectations on both sides, we get
\[
    &\E\|w_t - w^\star\|^2\\
    &\leq \E\left[\left\| \left((w_{t-1} \!-\! w^\star) \!-\! \diag(R_{t-1})\sum_{k=0}^{t-1}\beta_1^k G_{t-k-1}(w_{t-k-1} \!-\! w^\star)\right) \!-\! \left(\diag(R_{t-1})\sum_{k=0}^{t-1}\beta_1^k H_{t-k-1}(1\!-\!s_{t-k-1})\right) \right\|^2\right]\\
    &= \E\left[\left\| (w_{t-1} - w^\star) - \diag(R_{t-1})\sum_{k=0}^{t-1}\beta_1^k G_{t-k-1}(w_{t-k-1} - w^\star) \right\|^2\right] \\
    &\quad - 2\E\left[\left((w_{t-1} \!-\! w^\star) \!-\! \diag(R_{t-1})\sum_{k=0}^{t-1}\beta_1^k G_{t-k-1}(w_{t-k-1} \!-\! w^\star)\right)^\top \left(\diag(R_{t-1})\sum_{k=0}^{t-1}\beta_1^k H_{t-k-1}(1\!-\!s_{t-k-1})\right)\right]\\
    &\quad + \E\left[\left\| \diag(R_{t-1})\sum_{k=0}^{t-1}\beta_1^k H_{t-k-1}(1-s_{t-k-1}) \right\|^2\right]\\
    &= \E\left[\left\| (w_{t-1} \!-\! w^\star) \!-\! \diag(R_{t-1})\sum_{k=0}^{t-1}\beta_1^k G_{t-k-1}(w_{t-k-1} \!-\! w^\star) \right\|^2\right]
        + \E\left[\left\| \diag(R_{t-1})\sum_{k=0}^{t-1}\beta_1^k H_{t-k-1}(1\!-\!s_{t-k-1}) \right\|^2\right].
    \label{eq:explicit_update_1}
\]
In the above, the last equality follows due to unbiased subsampling, i.e., for all $t$, $\E[1-s_t] = 0$.
We now rewrite the first term in \cref{eq:explicit_update_1} as follows:
\[
    &\E\left[\left\| (w_{t-1} \!-\! w^\star) \!-\! \diag(R_{t-1})\sum_{k=0}^{t-1}\beta_1^k G_{t-k-1}(w_{t-k-1} \!-\! w^\star) \right\|^2\right]\\
    &=\E\left[ \left\| (w_{t-1} - w^\star) - \diag(R_{t-1})\sum_{k=0}^{t-1}\beta_1^k G_{t-k-1}(w_{t-k-1} - w^\star + w_{t-1} - w_{t-1}) \right\|^2 \right]\\
    &=\E\left[ \left\| (w_{t-1} - w^\star) - \diag(R_{t-1})\sum_{k=0}^{t-1}\beta_1^k G_{t-k-1}(w_{t-1} - w^\star) - \diag(R_{t-1})\sum_{k=0}^{t-1}\beta_1^k G_{t-k-1}(w_{t-k-1} - w_{t-1}) \right\|^2 \right]\\
    &=\E\left[ \left\| \left(I - \diag(R_{t-1})\sum_{k=0}^{t-1}\beta_1^k G_{t-k-1}\right)(w_{t-1} - w^\star) - \diag(R_{t-1})\sum_{k=1}^{t-1}\beta_1^k G_{t-k-1}\left( \sum_{j=1}^k \Delta_{t-j} \right) \right\|^2 \right],\label{eq:explicit_update}
\]
where $\Delta_{t-j}$ is as defined in \cref{eq:simple_notation}. The last equality above follows by rewriting $w_{t-k-1} - w_{t-1}$ as a telescoping sum.
Now let $A_t = \left(I - \diag(R_{t-1})\sum_{k=0}^{t-1}\beta_1^k G_{t-k-1}\right)$, $b_t= \diag(R_{t-1})\sum_{k=1}^{t-1}\beta_1^k G_{t-k-1}\left( \sum_{j=1}^k \Delta_{t-j} \right)$, and
    $c_t = \diag(R_{t-1})\sum_{k=0}^{t-1}\beta_1^k H_{t-k-1}(1-s_{t-k-1})$.
\cref{eq:explicit_update_1} then becomes
\[
    &\E \|w_t - w^\star\|^2\\
    &\leq \E \left[ \|A_t (w_{t-1} - w^\star) - b_t \|^2 \right] + \E\left[\| c_t \|^2 \right]\\
    &=\E \left[ \E \left[ (w_{t-1} - w^\star)^\top A_t^\top A_t (w_{t-1} - w^\star) - 2b_t^\top A_t (w_{t-1} - w^\star) + b_t^\top b_t \middle | w_{t-1}\right] \right] + \E\left[\| c_t \|^2 \right]\\
    &\leq\E \left[ (w_{t-1} - w^\star)^\top \E\left[ A_t^\top A_t \middle | w_{t-1} \right] (w_{t-1} - w^\star) + 2\left | \E\left[ b_t^\top A_t (w_{t-1} - w^\star) \middle | w_{t-1} \right] \right | + \E\left[ \|b_t\|^2 \middle | w_{t-1} \right] \right] + \E\left[\| c_t \|^2 \right]\\
    &\leq\E \left[ (w_{t-1} - w^\star)^\top \E\left[ A_t^\top A_t \middle | w_{t-1} \right] (w_{t-1} - w^\star)
        + 2 \sqrt{\E\left[ \|b_t\|^2 \middle | w_{t-1} \right]}
                    \sqrt{\E\left[ \|A_t (w_{t-1} - w^\star)\|^2 \middle | w_{t-1} \right]} \right]\\
        &\quad+ \E\left[\E\left[ \|b_t\|^2 \middle | w_{t-1} \right] \right] + \E\left[\| c_t \|^2 \right]\\
    &= \E \left[ (w_{t-1} \!-\! w^\star)^\top \E\left[ A_t^\top A_t \middle | w_{t-1} \right] (w_{t-1} \!-\! w^\star)
    \!+\! 2 \sqrt{\E\left[ \|b_t\|^2 \middle | w_{t-1} \right]}
                \sqrt{ (w_{t-1} \!-\! w^\star)^\top \E\left[A_t^\top A_t \middle | w_{t-1}\right] (w_{t-1} \!-\! w^\star)} \right]\\
    &\quad+ \E\left[\E\left[ \|b_t\|^2 \middle | w_{t-1} \right] \right] + \E\left[\| c_t \|^2 \right],
\]
where the last inequality is by Cauchy-Schwartz.
By \cref{lem:At,lem:bt,lem:ct}, we know that there exists $T^\star < \infty$ and $C_1, C_2>0$ such that for all $t > T^\star$, 
\[
    \E\left[ A_t^\top A_t \middle | w_{t-1} \right] \preceq \exp\left( -\frac{D}{\sqrt{t}} \right)I, \quad
    \E\left[ \|b_t\|^2 \middle | w_{t-1} \right] \leq \frac{C_1}{t^2},\quad
    \E\left[ \|c_t\|^2 \right] \leq \frac{C_2}{t}.
\]
Here $D=\frac{\lambda(1-\beta_1)r_\delta}{2\sqrt{\epsilon+(1-\beta_2)^{-1}U^2}}$ is as defined in \cref{lem:At}. 
We know $e^{-D/\sqrt{t}}\leq 1$. By \cref{assump:constraint}, we also have that for all $t \geq 1$, $\|w_{t-1}-w^\star\|^2 \leq \sum_{m=1}^M B^2 = MB^2$.
Therefore, 
\[
    \E \|w_t - w^\star\|^2 &\leq e^{-D/\sqrt{t}}\E\|w_{t-1} - w^\star\|^2 + 2\E\left[\sqrt{\frac{C_1}{t^2}}\sqrt{\exp\left( -\frac{D}{\sqrt{t}} \right) \|w_{t-1} - w^\star\|^2}\right] + \frac{C_1}{t^2} + \frac{C_2}{t}\\
    &\leq e^{-D/\sqrt{t}}\E\|w_{t-1} - w^\star\|^2 + 2\frac{B\sqrt{MC_1}}{t} + \frac{C_1}{t^2} + \frac{C_2}{t}\\
    &\leq e^{-D/\sqrt{t}}\E\|w_{t-1} - w^\star\|^2 + \frac{2B\sqrt{MC_1}+C_1+C_2}{t}.
\]
We unroll this recursion backward from $t$ to $T^\star$ to get
\[
    \E \|w_t - w^\star\|^2 &\leq e^{-D\sum_{\tau=T^\star+1}^{t}\frac{1}{\sqrt{\tau}}}\E\left[\left\| w_{T^\star}-w^\star \right\|^2\right] + \left( 2B\sqrt{MC_1}+C_1+C_2 \right) \sum_{\tau=T^\star+1}^{t} \frac{1}{\tau} e^{-D\sum_{u=\tau+1}^{t}\frac{1}{\sqrt{u}}}\\
    &\leq MB^2 e^{-D\sum_{\tau=T^\star+1}^{t}\frac{1}{\sqrt{\tau}}} + \left( 2B\sqrt{MC_1}+C_1+C_2 \right) \sum_{\tau=T^\star+1}^{t} \frac{1}{\tau} e^{-D\sum_{u=\tau+1}^{t}\frac{1}{\sqrt{u}}},
\]
where the last inequality again uses $\|w_{T^\star}-w^\star\|^2 \leq = MB^2$.
Since $\frac{1}{\sqrt{\tau}}$ monotonically decreases in $\tau$, we have that
$\sum_{\tau=T^\star+1}^{t} \frac{1}{\tau} \geq \int_{T^\star+1}^{t} \frac{1}{\sqrt{\tau}} d\tau = 2\left(\sqrt{t} - \sqrt{T^\star+1}\right)$.
Therefore, as $t\to\infty$,
\[
    \E \|w_t - w^\star\|^2 &\leq MB^2 e^{-2D(\sqrt{t} - \sqrt{T^\star+1})} + (2B\sqrt{MC_1}+C_1+C_2)\sum_{\tau=T^\star+1}^{t}\frac{1}{\tau}e^{-2D(\sqrt{t} - \sqrt{\tau+1})}\\
    &\leq MB^2e^{2D\sqrt{T^\star+1}} e^{-2D\sqrt{t}} + \left(2B\sqrt{MC_1}+C_1+C_2\right)e^{-2D\sqrt{t}}\sum_{\tau=1}^{t}\frac{1}{\tau}e^{2D\sqrt{\tau+1}}\\
    &= O\left( e^{-2D\sqrt{t}} + e^{-2D\sqrt{t}}\sum_{\tau=1}^{t}\frac{1}{\tau}e^{2D\sqrt{\tau+1}} \right). \label{eq:first_big_o}
\]
It is obvious that $e^{-2D\sqrt{t}} = O\left( \frac{1}{\sqrt{t}} \right)$ as $t \to\infty$. It remains to show that $e^{-2D\sqrt{t}}\sum_{\tau=1}^{t}\frac{1}{\tau}e^{2D\sqrt{\tau+1}} = O\left( \frac{1}{\sqrt{t}} \right)$ as $t\to\infty$.
We begin by noting that, since $\forall \tau \geq 1$, $\frac{\tau+1}{\tau} \leq 2$,
\[
    \sum_{\tau=1}^{t}\frac{1}{\tau}e^{2D\sqrt{\tau+1}} &= \sum_{\tau=1}^{t}\frac{1}{\tau}e^{2D\sqrt{\tau+1}}\frac{\tau}{\tau+1}\frac{\tau+1}{\tau}
    \leq 2\sum_{\tau=1}^{t}\frac{1}{\tau+1}e^{2D\sqrt{\tau+1}}
    = 2\sum_{\tau=2}^{t+1}\frac{1}{\tau}e^{2D\sqrt{\tau}}.
\]
We can then equivalently show $2e^{-2D\sqrt{t}}\sum_{\tau=2}^{t+1}\frac{1}{\tau}e^{2D\sqrt{\tau}} = O\left( \frac{1}{\sqrt{t}} \right)$ as $t\to\infty$.
We know that there exists $T' < \infty$ such that for all $\tau\geq T'$, $\frac{1}{\tau}e^{2D\sqrt{\tau}}$ monotonically increases with $\tau$. We therefore split the sum at $g(t)$, with $T' \leq g(t) \leq t$, to get 
\[
    2e^{-2D\sqrt{t}}\sum_{\tau=2}^{t+1}\frac{1}{\tau}e^{2D\sqrt{\tau}} &= 2e^{-2D\sqrt{t}}\sum_{\tau=2}^{g(t)-1}\frac{1}{\tau}e^{2D\sqrt{\tau}} + 2e^{-2D\sqrt{t}}\sum_{\tau=g(t)}^{t+1}\frac{1}{\tau}e^{2D\sqrt{\tau}}. \label{eq:split}
\]
We can bound the first term in \cref{eq:split} as follows
\[
    2e^{-2D\sqrt{t}}\sum_{\tau=2}^{g(t)-1}\frac{1}{\tau}e^{2D\sqrt{\tau}} \leq 2e^{2D\left( \sqrt{g(t)} - \sqrt{t} \right)}\sum_{\tau=2}^{g(t)-1}\frac{1}{\tau}
    \leq 2e^{2D\left( \sqrt{g(t)} - \sqrt{t} \right)}\left( \ln\left(g(t)\right) + 1 \right). \label{eq:condition_1}
\]

Looking at the second term in \cref{eq:split}, since $\tau \geq g(t)$ is large enough that the summand monotonically increases with $\tau$,
\[
    2e^{-2D\sqrt{t}}\sum_{\tau=g(t)}^{t+1}\frac{1}{\tau}e^{2D\sqrt{\tau}} \leq 2e^{-2D\sqrt{t}}\int_{g(t)}^{t+1}\frac{e^{2D\sqrt{\tau}}}{\tau}d\tau = 4e^{-2D\sqrt{t}}\int_{\sqrt{g(t)}}^{\sqrt{t+1}} \frac{e^{2Ds}}{s} ds, \label{eq:second_final_bound}
\]
Where the last equality follows by setting $s = \sqrt{\tau}$, $\tau = s^2$, $d\tau = 2s ds$.
Now for the integral in \cref{eq:second_final_bound}, we integrate by parts by defining $y = \frac{1}{s}$ and $dv = e^{2Ds}ds$:
\[
    \int_{\sqrt{g(t)}}^{\sqrt{t+1}} \frac{e^{2Ds}}{s} ds &= \frac{1}{2D\sqrt{t+1}}e^{2D\sqrt{t+1}} - \frac{1}{2D\sqrt{g(t)}}e^{2D\sqrt{g(t)}} + \frac{1}{2D}\int_{\sqrt{g(t)}}^{\sqrt{t+1}}\frac{e^{2Ds}}{s^2}ds\\
    &\leq \frac{1}{2D\sqrt{t+1}}e^{2D\sqrt{t+1}} + \frac{e^{2D\sqrt{t+1}}}{2D}\int_{\sqrt{g(t)}}^{\sqrt{t+1}}\frac{1}{s^2}ds\\
    &= \frac{1}{2D\sqrt{t+1}}e^{2D\sqrt{t+1}} + \frac{e^{2D\sqrt{t+1}}}{2D}\left( \frac{1}{\sqrt{g(t)}} - \frac{1}{\sqrt{t+1}} \right).
\]
Substituting the above back into \cref{eq:second_final_bound} to get
\[
    2e^{-2D\sqrt{t}}\sum_{\tau=g(t)}^{t+1}\frac{1}{\tau}e^{2D\sqrt{\tau}} &\leq 4e^{-2D\sqrt{t}}\left( \frac{1}{2D\sqrt{t+1}}e^{2D\sqrt{t+1}} + \frac{e^{2D\sqrt{t+1}}}{2D}\left( \frac{1}{\sqrt{g(t)}} - \frac{1}{\sqrt{t+1}} \right) \right)\\
    &= \frac{2}{D\sqrt{t+1}}e^{2D(\sqrt{t+1} - \sqrt{t})} + \frac{2}{D}e^{2D(\sqrt{t+1} - \sqrt{t})}\left( \frac{1}{\sqrt{g(t)}} - \frac{1}{\sqrt{t+1}} \right)\\
    &\leq \frac{2e^{2D}}{D\sqrt{t+1}} + \frac{2e^{2D}}{D}\left( \frac{1}{\sqrt{g(t)}} - \frac{1}{\sqrt{t+1}} \right). \label{eq:condition_2}
\]
Let $g(t) = \frac{t}{2}$. Then $T' \leq g(t) \leq t$ is satisfied for all $t \geq 2T'$. We can then combine \cref{eq:condition_1,eq:condition_2}, and have that for all $t \geq 2T'$,
\[
    2e^{-2D\sqrt{t}}\sum_{\tau=2}^{t+1}\frac{1}{\tau}e^{2D\sqrt{\tau}} &\leq 2e^{2D\left( \sqrt{t/2} - \sqrt{t} \right)}\left( \ln\left(\frac{t}{2}\right) + 1 \right) + \frac{2e^{2D}}{D\sqrt{t+1}} + \frac{2e^{2D}}{D}\left( \frac{1}{\sqrt{t/2}} - \frac{1}{\sqrt{t+1}} \right)\\
    &= 2e^{-2D\left( 1-\frac{1}{\sqrt{2}} \right)\sqrt{t}} \left( \ln t - \ln 2 + 1 \right) + \frac{2e^{2D}}{D\sqrt{t+1}} + \frac{4e^{2D}}{D\sqrt{t}},
\]
which is $O\left( \frac{1}{\sqrt{t}} \right)$ as $t\to\infty$.
Therefore, we arrive at the desired result that as $t\to\infty$,
\[
    \E \|w_t - w^\star\|^2 \leq O\left( e^{-2D\sqrt{t}} + 2e^{-2D\left( 1-\frac{1}{\sqrt{2}} \right)\sqrt{t}} \left( \ln t - \ln 2 + 1 \right) + \frac{2e^{2D}}{D\sqrt{t+1}} + \frac{4e^{2D}}{D\sqrt{t}} \right) = O\left( \frac{1}{\sqrt{t}} \right).
\]

\end{proof}

\subsection{Useful Lemmas}
We used several lemmas in the above proof of \cref{thm:convergence}. In this subsection, we present the proof of these lemmas.

\begin{lemma}\label{lem:At}
    Suppose \cref{assump:constraint,assump:exact,assump:grad_bound,assump:mixing,assump:noise} hold. 
    Define $D = \frac{\lambda(1-\beta_1)r_\delta}{2\sqrt{\epsilon+(1-\beta_2)^{-1}U^2}}$. 
    There exists $T<\infty$ such that $\forall t\geq T$,
    \[
        \E\left[ \left( I - \diag(R_{t-1})\sum_{k=0}^{t-1}\beta_1^k G_{t-k-1}\right)^\top\left( I - \diag(R_{t-1})\sum_{k=0}^{t-1}\beta_1^k G_{t-k-1}\right) \middle | w_{t-1} \right] \preceq \left(\exp\left( -\frac{D}{\sqrt{t}} \right)\right)I.
    \]
\end{lemma}

\begin{proof}[Proof of \cref{lem:At}]
    We begin by expanding the matrix product
    \[
        &\E\left[ \left( I - \diag(R_{t-1})\sum_{k=0}^{t-1}\beta_1^k G_{t-k-1}\right)^\top\left( I - \diag(R_{t-1})\sum_{k=0}^{t-1}\beta_1^k G_{t-k-1}\right) \middle | w_{t-1} \right]\\
        &= I - 2\E \left[ \diag(R_{t-1})\sum_{k=0}^{t-1}\beta_1^k G_{t-k-1} \middle | w_{t-1} \right] + 
            \E\left[ \left(\sum_{k=0}^{t-1}\beta_1^k G_{t-k-1}\right)^\top\diag(R^2_{t-1})\left(\sum_{k=0}^{t-1}\beta_1^k G_{t-k-1}\right) \middle | w_{t-1} \right].\label{eq:first_lemma}
    \]
    We bound the first expectation from below and the second expectation from above.
    
    We being by bounding $\E \left[ \diag(R_{t-1})\sum_{k=0}^{t-1}\beta_1^k G_{t-k-1} \middle | w_{t-1} \right]$ from below. 
    Following the update rule as specified in \cref{sec:update_rule}, we expand the $i^{\text{th}}$ entry of $R_{t-1}$ as defined by \cref{eq:simple_notation} and get
    \[
        R_{t-1,i} = \frac{\alpha_{t,i}}{\sqrt{t(\epsilon + v_{t,i})}} 
        = \left(\frac{1-\beta_1}{1-\beta_1^t}\right)\left(\frac{\sqrt{1-\beta_2^t}}{\sqrt{1-\beta_2}}\right)\frac{1}{1-\beta_1^t}\left( (1-\beta_1)\left( \sum_{k=0}^{t-1}\beta_1^k \bar{r}_{t-k,i} \right) + 
            \beta_1^t\bar{r}_{0,i} \right)\frac{1}{\sqrt{t(\epsilon + v_{t,i})}}. \label{eq:expanded_R}
    \]
    By \cref{assump:grad_bound} and that $|\beta_2| < 1$, we can bound $v_{t,i}$ as defined in \cref{eq:v_update} by
    \[
        v_{t,i} = \sum_{k=0}^{t-1}\beta_2^k g_i^2(w_{t-k-1}, \theta_{t-k-1}) \leq U^2\sum_{k=0}^{t-1}\beta_2^k 
        \leq U^2(1-\beta_2)^{-1}.
    \]
    Together with $|\beta_1|<1$ and that $\forall t, i$, $\bar{r}_{t,i} \geq r_\delta$, 
    \[
        R_{t-1,i} \geq 
        \left(\frac{1-\beta_1}{1-\beta_1^t}\right)\left(\frac{\sqrt{1-\beta_2^t}}{\sqrt{1-\beta_2}}\right)\frac{r_\delta}{1-\beta_1^t}\frac{1}{\sqrt{t}\sqrt{\epsilon + (1-\beta_2)^{-1}U^2}}
        \geq \frac{(1-\beta_1)r_\delta}{\sqrt{t}\sqrt{\epsilon + (1-\beta_2)^{-1}U^2}}.
    \]
    As a result, for all $t$, we have that $\diag(R_{t-1}) \succeq \frac{(1-\beta_1)r_\delta}{\sqrt{t}\sqrt{\epsilon+(1-\beta_2)^{-1}U^2}}I$.

    Now let $A=\diag(R_{t-1}) - \frac{1}{2}\left(\min_{1\leq i\leq M} R_{t-1,i}\right)I$. We know A is diagonal and $A \succeq \frac{(1-\beta_1)r_\delta}{2\sqrt{t}\sqrt{\epsilon+(1-\beta_2)^{-1}U^2}}I$.
    We also know $Q \coloneqq \sum_{k=0}^{t-1}\beta_1^k G_{t-k-1} \succeq 0$ since $G_t$ are sample covariance matrices. 
    Together, using $\Lambda_{\min}$ to denote the minimum eigenvalue, we have
        $\Lambda_{\min}\left( AQ \right) = \Lambda_{\min}\left( A^\frac{1}{2}QA^\frac{1}{2} \right) \geq 0$, and so $AQ \succeq 0$.
    Therefore,
    \[
        \diag(R_{t-1})\sum_{k=0}^{t-1}\beta_1^k G_{t-k-1} \succeq \frac{1}{2}\left(\min_{1\leq i\leq M} R_{t-1,i}\right)\sum_{k=0}^{t-1}\beta_1^k G_{t-k-1} = \frac{(1-\beta_1)r_\delta}{2\sqrt{t}\sqrt{\epsilon+(1-\beta_2)^{-1}U^2}}\sum_{k=0}^{t-1}\beta_1^k G_{t-k-1}.
    \]
    Using the above, we have that
    \[
        \E \left[ \diag(R_{t-1})\sum_{k=0}^{t-1}\beta_1^k G_{t-k-1} \middle | w_{t-1} \right] 
        &\succeq \frac{(1-\beta_1)r_\delta}{2\sqrt{t}\sqrt{\epsilon+(1-\beta_2)^{-1}U^2}}\sum_{k=0}^{t-1}\beta_1^k \E\left[G_{t-k-1} \middle | w_{t-1}\right]\\
        &= \frac{(1-\beta_1)r_\delta}{2\sqrt{t}\sqrt{\epsilon+(1-\beta_2)^{-1}U^2}}\sum_{k=0}^{t-1}\beta_1^k 
            \E\left[\E\left[ G_{t-k-1} \middle| w_{t-k-1}, \theta_{t-k-2} \right] \middle | w_{t-1}\right]\\
        &\succeq \frac{\lambda(1-\beta_1)r_\delta}{2\sqrt{t}\sqrt{\epsilon+(1-\beta_2)^{-1}U^2}}\left(\sum_{k=0}^{t-1}\beta_1^k\right)I\\
        &\succeq \frac{\lambda(1-\beta_1)r_\delta}{2\sqrt{t}\sqrt{\epsilon+(1-\beta_2)^{-1}U^2}}I, \label{eq:first_term}
    \]
    where the inequalities are due to \cref{assump:mixing} and $|\beta_1|<1$.

    We now bound $\E\left[ \left(\sum_{k=0}^{t-1}\beta_1^k G_{t-k-1}\right)^\top\diag(R^2_{t-1})\left(\sum_{k=0}^{t-1}\beta_1^k G_{t-k-1}\right) \middle | w_{t-1} \right]$. We similarly begin 
    by bounding $R_{t-1,i}$ from the other direction. By \cref{assump:constraint}, $\bar{r}_{t,i} \leq B$. Together with $v_{t,i}\geq 0$, and that $|\beta_1|<1, |\beta_2|<1$, 
    we can bound \cref{eq:expanded_R} from above by
    \[
        R_{t-1,i} 
        \leq \left(\frac{1-\beta_1}{1-\beta_1^t}\right)\left(\frac{\sqrt{1-\beta_2^t}}{\sqrt{1-\beta_2}}\right)\frac{B}{1-\beta_1^t} \frac{1}{\sqrt{t}\sqrt{\epsilon}}
        \leq \frac{B}{t\epsilon(1-\beta_1)\sqrt{1-\beta_2}}.\label{eq:R_ub}
    \]
    Again using $|\beta_1|<1, |\beta_2|<1$, and squaring $R_{t-1,i}$, we have that
    \[
        \diag(R_{t-1}^2) \preceq \frac{B^2}{t\epsilon(1-\beta_1)^2(1-\beta_2)}I. \label{eq:simple_R_ub}
    \]
    Therefore,
    \[
        &\E\left[ \left(\sum_{k=0}^{t-1}\beta_1^k G_{t-k-1}\right)^\top\diag(R^2_{t-1})\left(\sum_{k=0}^{t-1}\beta_1^k G_{t-k-1}\right) \middle | w_{t-1} \right]\\
        &\preceq \frac{B^2}{t\epsilon(1-\beta_1)^2(1-\beta_2)} \E\left[ \left(\sum_{k=0}^{t-1}\beta_1^k G_{t-k-1}\right)^\top\left(\sum_{k=0}^{t-1}\beta_1^k G_{t-k-1}\right) \middle | w_{t-1} \right]\\
        &= \frac{B^2}{t\epsilon(1-\beta_1)^2(1-\beta_2)} \sum_{k=0}^{t-1}\sum_{k'=0}^{t-1}\beta_1^k \beta_1^{k'} \E\left[ G_{t-k-1}^\top G_{t-k'-1} \middle | w_{t-1}\right]\\
        &= \frac{B^2}{t\epsilon(1-\beta_1)^2(1-\beta_2)} \sum_{k=0}^{t-1}\sum_{k'=0}^{t-1}\beta_1^k \beta_1^{k'} 
        \E\left[ \E\left[ G_{t-k'-1}^\top G_{t-k-1} \middle | w_{t-k'-1}, \theta_{t-k'-2}, w_{t-k-1}, \theta_{t-k-2} \right] \middle | w_{t-1}\right]\\
        &\preceq \frac{\bar{\lambda}B^2}{t\epsilon(1-\beta_1)^2(1-\beta_2)}\left(\sum_{k=0}^{t-1}\beta_1^k\right)^2I\\
        &\preceq \frac{\bar{\lambda}B^2}{t\epsilon(1-\beta_1)^4(1-\beta_2)}I, \label{eq:second_term}
    \]  
    where the second last inequality is due to \cref{assump:noise}, and the last inequality is by $|\beta_1|<1$.
    Let $D' = \frac{\bar{\lambda}B^2}{\epsilon(1-\beta_1)^4(1-\beta_2)}$ and recall that $D = \frac{\lambda(1-\beta_1)r_\delta}{2\sqrt{\epsilon+(1-\beta_2)^{-1}U^2}}$.
    Together by \cref{eq:first_term,eq:second_term}, we can bound \cref{eq:first_lemma} by 
    \[
        \E\left[ \left( I - \diag(R_{t-1})\sum_{k=0}^{t-1}\beta_1^k G_{t-k-1}\right)^\top\left( I - \diag(R_{t-1})\sum_{k=0}^{t-1}\beta_1^k G_{t-k-1}\right) \middle | w_{t-1} \right]
        \preceq \left(1 - \frac{2D}{\sqrt{t}} + \frac{D'}{t}\right)I.
    \]
    Since $D, D' > 0$, we have for all $t\geq \frac{D'^2}{D^2}$, 
    $1 - \frac{2D}{\sqrt{t}} + \frac{D'}{t} \leq 1 - \frac{D}{\sqrt{t}} \leq \exp\left(-{\frac{D}{\sqrt{t}}}\right)$.
    Therefore, for all $t\geq \frac{D_2^2}{D_1^2}$, we have that
    \[
        \E\left[ \left( I - \diag(R_{t-1})\sum_{k=0}^{t-1}\beta_1^k G_{t-k-1}\right)^\top\left( I - \diag(R_{t-1})\sum_{k=0}^{t-1}\beta_1^k G_{t-k-1}\right) \middle | w_{t-1} \right]
        \preceq \left(\exp\left( -\frac{D}{\sqrt{t}} \right)\right)I.
    \]
\end{proof}

\begin{lemma}\label{lem:bt}
    Suppose \cref{assump:constraint,assump:exact,assump:grad_bound,assump:mixing,assump:noise} hold. We have that as $t \to \infty$,
    \[
        \E\left[\left\| \diag(R_{t-1})\sum_{k=1}^{t-1}\beta_1^k G_{t-k-1}\left( \sum_{j=1}^k \Delta_{t-j} \right)\right\|^2 \middle | w_{t-1}\right] = O\left(\frac{1}{t^2}\right).
    \]
\end{lemma}

\begin{proof}[Proof of \cref{lem:bt}]
    We begin by expanding the norm
    \[
        &\E\left[\left\| \diag(R_{t-1})\sum_{k=1}^{t-1}\beta_1^k G_{t-k-1}\left( \sum_{j=1}^k \Delta_{t-j} \right)\right\|^2 \middle | w_{t-1}\right]\\
        &\leq \E\left[ \left(\max_{1\leq i\leq M}R_{t-1,i}\right)^2 \left\| \sum_{k=1}^{t-1}\beta_1^k G_{t-k-1}\left( \sum_{j=1}^k \Delta_{t-j} \right)\right\|^2 \middle | w_{t-1}\right]\\
        &\leq \frac{B^2}{t\epsilon(1-\beta_1)^2(1-\beta_2)} \E\left[\sum_{k,k'=1}^{t-1} \beta_{1}^{k+k'} \left(\sum_{j=1}^k \Delta_{t-j} \right)^\top G_{t-k-1}^\top G_{t-k'-1}  \left(\sum_{j=1}^{k'} \Delta_{t-j} \right) \middle | w_{t-1}\right] \\
        &= \frac{B^2}{t\epsilon(1-\beta_1)^2(1-\beta_2)} \E\left[\sum_{k,k'=1}^{t-1} \beta_{1}^{k+k'} \left(w_{t-k-1}-w_{t-1} \right)^\top G_{t-k-1}^\top G_{t-k'-1}  \left(w_{t-k'-1}-w_{t-1} \right) \middle | w_{t-1}\right],\label{eq:intermediate}
    \]
    where the second inequality is by \cref{eq:simple_R_ub}, and the last equality follows after writing $w_{t-k-1}-w_{t-1}$ as a telescoping sum. 
    Using \cref{assump:noise}, we can bound the expectation in \cref{eq:intermediate} as follows:
    \[
        &\E\left[\sum_{k,k'=1}^{t-1} \beta_{1}^{k+k'} \left(w_{t-k-1}-w_{t-1} \right)^\top G_{t-k-1}^\top G_{t-k'-1}  \left(w_{t-k'-1}-w_{t-1} \right) \middle | w_{t-1}\right]\\
        &= \E\left[\sum_{k,k'=1}^{t-1} \beta_{1}^{k\!+\!k'} \left(w_{t-k-1}\!-\!w_{t-1} \right)^\top \E\left[G_{t-k-1}^\top G_{t-k'-1} \middle | w_{t-k-1}, \theta_{t-k-2}, w_{t-k'-1}, \theta_{t-k'-2} \right]
            \left(w_{t-k'-1}\!-\!w_{t-1} \right) \middle | w_{t-1}\right]\\
        &\leq \bar{\lambda}\E\left[\sum_{k,k'=1}^{t-1} \beta_{1}^{k\!+\!k'} \left(w_{t-k-1}\!-\!w_{t-1} \right)^\top \left(w_{t-k'-1}\!-\!w_{t-1} \right) \middle | w_{t-1}\right]\\
        &= \bar{\lambda}\E\left[\sum_{k,k'=1}^{t-1} \beta_{1}^{k\!+\!k'} \left( \sum_{j=1}^k \Delta_{t-j} \right)^\top \left( \sum_{j=1}^{k'} \Delta_{t-j} \right) \middle | w_{t-1}\right].
    \]
    Therefore,
    \[
        &\E\left[\left\| \diag(R_{t-1})\sum_{k=1}^{t-1}\beta_1^k G_{t-k-1}\left( \sum_{j=1}^k \Delta_{t-j} \right)\right\|^2 \middle | w_{t-1} \right]\\
        &\leq \frac{\bar{\lambda}B^2}{t\epsilon(1-\beta_1)^2(1-\beta_2)}\E\left[\sum_{k,k'=1}^{t-1} \beta_{1}^{k\!+\!k'} \left( \sum_{j=1}^k \Delta_{t-j} \right)^\top \left( \sum_{j=1}^{k'} \Delta_{t-j} \right) \middle | w_{t-1}\right]\\
        &\leq \frac{\bar{\lambda}B^2}{t\epsilon(1-\beta_1)^2(1-\beta_2)}\E\left[\sum_{k,k'=1}^{t-1} \beta_{1}^{k+k'} \left\| \sum_{j=1}^k \Delta_{t-j} \right\| \left\| \sum_{j=1}^{k'} \Delta_{t-j} \right\| \middle | w_{t-1} \right].
    \]
    We now bound $\left\| \sum_{j=1}^k \Delta_{t-j} \right\|^2$.
    \[
        \left\| \sum_{j=1}^k \Delta_{t-j} \right\|^2 
        = \sum_{j,j'=1}^k \Delta_{t-j}^\top\Delta_{t-j'}
        \leq \sum_{j,j'=1}^k \|\Delta_{t-j}\|\|\Delta_{t-j'}\|.
    \]  
    By \cref{eq:simple_notation,eq:m_update}, we can write
    \[
        \|\Delta_{t-j}\|^2 &= \sum_{i=1}^M R_{n-j,i}^2 m^2_{n-j,i}\\
        &\leq \frac{B^2}{(t-j)\epsilon(1-\beta_1)^2(1-\beta_2)} \sum_{i=1}^M m_{n-j,i}^2\\
        &= \frac{B^2}{(t-j)\epsilon(1-\beta_1)^2(1-\beta_2)} \sum_{i=1}^M \left( \sum_{k=0}^{t-j-1}\beta_1^k g_i(w_{t-j-k-1}, \theta_{t-j-k-1}) \right)^2\\
        &\leq \frac{U^2B^2}{(t-j)\epsilon(1-\beta_1)^2(1-\beta_2)} \sum_{i=1}^M \left( \sum_{k=0}^{t-j-1}\beta_1^k \right)^2\\
        &\leq \frac{U^2B^2M}{(t-j)\epsilon(1-\beta_1)^4(1-\beta_2)},
    \]  
    where the first inequality is by \cref{eq:R_ub}, and the second inequality by \cref{assump:grad_bound}, and the last inequality by $|\beta_1|<1$.
    Let $D_1 = \frac{\bar{\lambda}B^2}{\epsilon(1-\beta_1)^2(1-\beta_2)}$, $D_2 = \frac{U^2B^2M}{\epsilon(1-\beta_1)^4(1-\beta_2)}$, we have
    \[
        &\E\left[\left\| \diag(R_{t-1})\sum_{k=1}^{t-1}\beta_1^k G_{t-k-1}\left( \sum_{j=1}^k \Delta_{t-j} \right)\right\|^2 \middle | w_{t-1}\right]\\
        &\leq \frac{D_1}{t}\sum_{k,k'=1}^{t-1}\beta_{1}^{k+k'}\sum_{j,j'=1}^{k}\frac{\sqrt{D_2}}{\sqrt{t-j}}\frac{\sqrt{D_2}}{\sqrt{t-j'}}
        = \frac{D_1D_2}{t}\left( \sum_{k=1}^{t-1}\beta_1^k \sum_{j=1}^{k}\frac{1}{\sqrt{t-j}} \right)^2.
    \]
    If we can show that, as $t\to\infty$, $S(t)\coloneqq \sum_{k=1}^{t-1}\beta_1^k\sum_{j=1}^k \frac{1}{\sqrt{t-j}} = O\left( \frac{1}{\sqrt{t}} \right)$, 
    then we have, as $t\to\infty$, $\frac{D_1D_2}{t}\left( \sum_{k=1}^{t-1}\beta_1^k \sum_{j=1}^{k}\frac{1}{\sqrt{t-j}} \right)^2 = O\left(\frac{1}{t^2}\right)$,
    thus concluding the proof.
    We now show that $S(t) = O\left( \frac{1}{\sqrt{t}} \right)$ as $t\to\infty$.
    \[
        S(t) = \sum_{k=1}^{t-1}\beta_1^k \sum_{j=1}^{k}\frac{1}{\sqrt{t-j}}
        = \sum_{j=1}^{t-1}\sum_{k=j}^{t-1}\beta_1^k \frac{1}{\sqrt{t-j}}
        = \sum_{j=1}^{t-1}\frac{1}{\sqrt{t-j}} \sum_{k=j}^{t-1}\beta_1^k
        = \sum_{j=1}^{t-1}\frac{1}{\sqrt{t-j}} \frac{\beta_1^j(1-\beta_1^{t-j})}{1-\beta_1}.
    \]
    We decompose the above into two sums to get
    \[
        S(t)
        = \frac{1}{1\!-\!\beta_1}\sum_{j=1}^{t-1}\frac{\beta^j_1}{\sqrt{t-j}} - \frac{1}{1-\beta_1}\sum_{j=1}^{t-1}\frac{\beta_1^j\beta_1^{t-j}}{\sqrt{t-j}}
        = \frac{1}{1\!-\!\beta_1}\sum_{j=1}^{t-1}\frac{\beta^j_1}{\sqrt{t-j}} - \frac{\beta_1^t}{1-\beta_1}\sum_{j=1}^{t-1}\frac{1}{\sqrt{t-j}}
        \leq \frac{1}{1-\beta_1}\sum_{j=1}^{t-1}\frac{\beta^j_1}{\sqrt{t-j}}.
    \]
    Splitting the sum above at $\lfloor t/2 \rfloor$, we get that
    \[
        S(t) = \frac{1}{1-\beta_1}\sum_{j=1}^{\lfloor t/2 \rfloor} \frac{\beta_1^j}{\sqrt{t-j}} + 
                \frac{1}{1-\beta_1}\sum_{j=\lfloor t/2 \rfloor +1}^{t-1} \frac{\beta_1^j}{\sqrt{t-j}}.
    \]
    In the first sum, since $j\leq \lfloor t/2 \rfloor$, we know $t-j\geq t - \frac{t}{2} = \frac{t}{2}$. Then 
    \[
        \frac{1}{1-\beta_1}\sum_{j=1}^{\lfloor t/2 \rfloor} \frac{\beta_1^j}{\sqrt{t-j}} \leq \frac{1}{1-\beta_1}\frac{\sqrt{2}}{\sqrt{t}}\sum_{j=1}^{\lfloor t/2 \rfloor}\beta_1^j \leq \frac{\beta_1\sqrt{2}}{(1-\beta_1)^2\sqrt{t}}.
    \]
    In the second sum, since $\lfloor t/2 \rfloor + 1 \leq j \leq t-1$, we know $t-j\geq 1$. Then
    \[
        \frac{1}{1-\beta_1}\sum_{j=\lfloor t/2 \rfloor +1}^{t-1} \frac{\beta_1^j}{\sqrt{t-j}} \leq \frac{1}{1-\beta_1} \sum_{j=\lfloor t/2 \rfloor +1}^{t-1}\beta_1^j \leq 
        \frac{1}{1-\beta_1} \sum_{j=\lfloor t/2 \rfloor +1}^{\infty}\beta_1^j \leq \frac{\beta_1^{\lfloor t/2 \rfloor +1}}{(1-\beta_1)^2}.
    \]
    Since $|\beta_1|<1$, $\beta_1^{\lfloor t/2 \rfloor +1}$ decays faster than $\frac{1}{\sqrt{t}}$ as $t\to\infty$. 
    Therefore, we have that, as $t\to\infty$, $S(t) = O\left( \frac{1}{\sqrt{t}} \right)$.
\end{proof}

\begin{lemma}\label{lem:ct}
    Suppose \cref{assump:constraint,assump:exact,assump:grad_bound,assump:mixing,assump:noise} hold. We have that as $t \to \infty$,
    \[
        \E\left[\left\| \diag(R_{t-1})\sum_{k=0}^{t-1}\beta_1^k H_{t-k-1}(1-s_{t-k-1}) \right\|^2\right] = O\left( \frac{1}{t} \right).
    \]
\end{lemma}

\begin{proof}[Proof of \cref{lem:ct}]
    We begin by expanding the norm
    \[
        &\E\left[\left\| \diag(R_{t-1})\sum_{k=0}^{t-1}\beta_1^k H_{t-k-1}(1-s_{t-k-1}) \right\|^2\right]\\
        &= \sum_{k,k'=0}^{t-1}\beta_{1}^{k+k'} \E\left[ \left(H_{t-k-1}(1-s_{t-k-1})\right)^\top
                                                    \diag(R_{t-1}^2)
                                                    \left(H_{t-k'-1}(1-s_{t-k'-1})\right)
        \right]\\
        &\leq \frac{B^2}{t\epsilon(1-\beta_1)^2(1-\beta_2)} \sum_{k,k'=0}^{t-1}\beta_1^{k+k'}\E\left[ (1-s_{t-k-1})^\top H_{t-k-1}^\top H_{t-k'-1} (1-s_{t-k'-1}) \right]\\
        &= \frac{B^2}{t\epsilon(1-\beta_1)^2(1-\beta_2)} \sum_{k=0}^{t-1}\beta_1^{2k}\E\left[ (1-s_{t-k-1})^\top H_{t-k-1}^\top H_{t-k-1} (1-s_{t-k-1}) \right]\\
        &= \frac{B^2}{t\epsilon(1-\beta_1)^2(1-\beta_2)} \sum_{k=0}^{t-1}\beta_1^{2k} \E\left[\left\|H_{t-k-1} (1-s_{t-k-1})\right\|^2\right].
    \]
    In the above, the inequality is by \cref{eq:simple_R_ub}; the second last equality is due to unbiased subsampling and that when $k\neq k'$, $s_{t-k-1} \indep s_{t-k'-1}$.
    If we can show $\forall t$, $\E\left[\left\|H_{t} (1-s_{t})\right\|^2\right]$ is uniformly bounded above by some constant $C$, then we have
    \[
        &\E\left[\left\| \diag(R_{t-1})\sum_{k=0}^{t-1}\beta_1^k H_{t-k-1}(1-s_{t-k-1}) \right\|^2\right]\\
        &\leq \frac{CB^2}{t\epsilon(1-\beta_1)^2(1-\beta_2)} \sum_{k=0}^{t-1} \beta_1^{2k}\\
        &\leq \frac{CB^2}{t\epsilon(1-\beta_1)^2(1-\beta_2)}\frac{1}{1-\beta_1^2},
    \]
    where the last line is by $|\beta_1|<1$. We can therefore conclude as $t\to\infty$,
    $\E\!\left[\left\| \diag(\!R_{t-1}\!)\sum_{k=0}^{t-1}\beta_1^k H_{t-k-1}(\!1\!-\!s_{t\!-\!k\!-\!1}\!) \right\|^2\right] = O\left(\frac{1}{t}\right)$.
    It now remains to show that $\forall t$, $\E\left[\left\|H_{t} (1-s_{t})\right\|^2\right]$ is uniformly bounded above by a constant.

    By \cref{eq:gradient}, we have that
    \[
        \E\left[\left\|g(w_t, \theta_t, \mathcal{S}_t)\right\|^2\right]
        &= \E\left[ \left\| G_t (w_t - w^\star) \right\|^2 + \left\| H_t(1-s_t) \right\|^2 + 2(w_{t}-w^\star)^\top G_{t}^\top H_t(1-s_t) \right]\\
        &= \E\left[ \left\| G_t (w_t - w^\star) \right\|^2 + \left\| H_t(1-s_t) \right\|^2 + 2(w_{t}-w^\star)^\top G_{t}^\top H_t\E\left[(1-s_t) \middle | w_t, \theta_t \right] \right]\\
        &= \E\left[ \left\| G_t (w_t - w^\star) \right\|^2 + \left\| H_t(1-s_t) \right\|^2\right],
    \]
    where the last equality is due to unbiased subsampling. Together with \cref{assump:grad_bound}, we have
    \[
        \E\left[ \left\| H_t(1-s_t) \right\|^2 \right] &\leq \E\left[ \left\| G_t (w_t - w^\star) \right\|^2 + \left\| H_t(1-s_t) \right\|^2\right]
        \leq \E\left[\left\|g(w_t, \theta_t, \mathcal{S}_t)\right\|^2\right]
        \leq MU^2,
    \]
    thus concluding the proof.
\end{proof}

\end{document}